\documentclass[]{interact}

\usepackage[caption=false]{subfig}

\usepackage{amsmath}
\usepackage{savesym}
\usepackage{amssymb}
\usepackage{graphicx}
\usepackage{xfrac}
\usepackage{setspace}
\usepackage{romannum}
\usepackage{color}
\usepackage{cancel}
\usepackage{algorithm,algorithmic}
\newcommand*{\myDots}{\ifmmode\mathellipsis\else.\kern-0.05em.\kern-0.05em.\fi}
\usepackage{arydshln}
\usepackage{tikz}
\usetikzlibrary{shapes,arrows,positioning,calc}
\tikzset{
block/.style = {draw, fill=white, rectangle, minimum height=2.5em, minimum width=3em},
tmp/.style = {coordinate},
sum/.style= {draw, fill=white, circle, node distance=1cm},
input/.style = {coordinate},
output/.style= {coordinate},
pinstyle/.style = {pin edge={to-,thin,black}}}

\usepackage{amsthm}

\newtheorem{theorem}{Theorem}

\theoremstyle{definition}

\newtheorem{assumption}{A.}

\newtheorem{remark}{Remark}
\theoremstyle{plain}

\theoremstyle{definition}

\theoremstyle{remark}

\usepackage[utf8]{inputenc}
\usepackage[english]{babel}

\usepackage[
backend=biber,
style=apa,
]{biblatex}

\begin{document}

\title{Decoupled Reference Governors: A Constraint Management Technique for MIMO Systems }

 \author{
\name{Yudan Liu\textsuperscript{a}\thanks{CONTACT Yudan Liu. Email: yliu38@uvm.edu}, 
Joycer Osorio\textsuperscript{b}\thanks{CONTACT Joycer Osorio. Email: Joycer.Osorio@uvm.edu},
Hamid R. Ossareh\textsuperscript{c}\thanks{CONTACT Hamid R. Ossareh. Email: Hamid.Ossareh@uvm.edu}}
\affil{\textsuperscript{a,b,c}Department of Electrical and Biomedical Engineering, University Of Vermont, Burlington, VT USA}
}

\maketitle

\begin{abstract}
This paper presents a computationally efficient solution for constraint management of  multi-input and multi-output (MIMO) systems.  The solution, referred to as the Decoupled Reference Governor (DRG), maintains the highly-attractive computational features of Scalar Reference Governors (SRG) while having performance comparable to Vector Reference Governors (VRG). DRG is based on decoupling the input-output dynamics of the system, followed by the deployment of a bank of SRGs for each decoupled channel. We present two formulations of DRG: DRG-tf, which is based on system decoupling using transfer functions, and DRG-ss, which is built on state feedback decoupling.  A detailed set-theoretic analysis of DRG, which highlights its main characteristics, is presented. We also show a quantitative comparison  between DRG and the VRG to illustrate the computational advantages of DRG. The robustness of this approach to disturbances and uncertainties is also investigated.
\end{abstract}

\begin{keywords}
Constraint management; Reference governors; Maximal admissible set; System decoupling; MIMO systems
\end{keywords}

\section{Introduction}\label{sec: intro}

Control and constraint management of systems with multiple inputs and multiple outputs (i.e., MIMO systems)  have been studied in the field of controls for many decades. The control of MIMO systems has  been the focus of many works in the literature, for example the Linear Quadratic Regulator (LQR), 
state feedback control methods,
sliding mode control,  
 $\mathcal{H}_2$ and $\mathcal{H}_\infty$ control, and decentralized and centralized control methods, please see \cite{Zhou_1996, Ge_2014, Garelli_2006, Burl_1998, Skogestad_2007} and the references therein. The problem of constraint management of  MIMO systems  has been explored as well. One route is to first find a suitable compensator to decouple the input-output dynamics, see \cite{Skogestad_2007, MacFarlane_1970,Macfarlane_1983}.
Afterwards, a diagonal controller for the newly decoupled plant is designed. The constraint management part is handled by nonlinear functions (e.g., saturation functions) that maintain the constrained signal within the desired bounds (\cite{Aastrom_1995}). However, this approach can compromise the closed-loop stability and may not enforce state constraints.  Another approach is Model Predictive Control (MPC), see \cite{Shah_2011,Bemporad_2002}, which addresses both tracking and  constraint management at the same time. This approach for constraint management in MIMO systems is explored in works like  \cite{Elliott_2013}, where  decentralized MPC strategies are proposed. 
Other MPC solutions are centralized (\cite{Wang_2003}), distributed (\cite{Camponogara_2002}), and cascade or hierarchical strategies (\cite{Scattolini_2009}). However, MPC tends to be computationally demanding, which has limited its applicability, especially for systems with fast dynamics and/or high order. Theoretical guarantees such as stability are also difficult to obtain in practice. Other approaches to solve constraint management are $l_1$-optimal control, see \cite{mcdonald1991ℓ1}, barrier Lyapunov function, see \cite{tee2009barrier}, and constrained LQR, see \cite{scokaert1998constrained}.

A relatively new  constraint management technique, which can be designed independently of the tracking controller and  alleviates the above shortcomings of MPC, is the Reference Governor (\cite{Kolmanovsky_2014,GARONE2017306}), also referred to as the  Scalar Reference Governor (SRG). It is  an add-on scheme for enforcing pointwise-in-time state and control constraints by modifying, whenever  required, the reference to a well-designed stable closed-loop system. A block diagram of SRG is shown in Figure~\ref{fig: Governor scheme}, where $y(t)$ is the constrained output, $r(t)$ is the reference, $v(t)$ is the governed reference, and $x(t)$ is the system state (measured or estimated). 
To compute $v(t)$, SRG employs the so-called maximal admissible set (MAS) (\cite{Gilbert_1991}),  which is defined as the set of all inputs and initial conditions that are constraint-admissible. By solving a simple linear program over this set, SRG selects a $v(t)$ that is as close as possible to $r(t)$ such that the constraints are satisfied for all time.
\begin{figure}
\centering
\begin{tikzpicture}[auto, node distance=1.5cm,>=latex']
\node [input, name=rinput] (rinput) {};
\node [block,minimum size=1.5cm, right of=rinput,text width=1.5cm,align=center] (controller) {Reference \hbox{Governor}};
\node [block,minimum size=1.5cm, right of=controller,node distance=3cm,text width=2cm,align=center] (system)
{\hbox{Closed-Loop} Plant};
\node [output, right of=system, node distance=2cm] (output) {};
\node [tmp, below of=controller,node distance=1.2cm] (tmp1){$s$};
\draw [->] (rinput) -- node{\hspace{-0.4cm}$r(t)$} (controller);

\draw [->] (controller) -- node [name=v]{$v(t)$}(system);
\draw [->] (system) -- node [name=y] {$y(t)$}(output);
\draw [->] (system) |- (tmp1)-| node[pos=0.75] {$x(t)$} (controller);
\end{tikzpicture}
\caption{Scalar reference governor block diagram} \label{fig:RGblock}
 \label{fig: Governor scheme}
\end{figure}
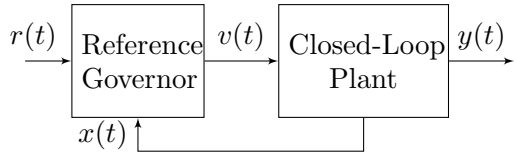

Standard SRG uses a single decision variable in the linear program to simultaneously govern all the channels of a MIMO system. As a result, it tends to have a conservative response. A modification of the SRG, which performs well in MIMO systems, is the so-called Vector Reference Governor (VRG), see \cite{GARONE2017306}. This technique handles constraint management by solving a quadratic program (QP) with multiple decision variables (one for each reference input). Even though VRG shares some properties with SRG, its implementation demands a higher computational load in comparison with SRG. This is because of the QP with multiple decision variables that must be solved at each time step, either by implicit methods or  multi-parametric explicit methods. In this paper, we  present a new reference governor solution for MIMO systems that maintains the computational simplicity of the SRG, but with performance similar to VRG. The solution, referred to as the Decoupled Reference Governor (DRG), is based on decoupling the input-output dynamics of the system, followed by the deployment of a bank of SRGs for each decoupled channel. Since the decoupling operation can be performed in both transfer function and state-space domains, we investigate two DRG formulations: DRG-tf, and DRG-ss, as summarized next.

 The block diagram of the DRG-tf method is shown in Figure~\ref{fig:Decoupled with RG}, where $G(z)$ is the closed-loop system with inputs $u_i$ and constrained outputs $y_i$. In this block diagram, we have assumed that the system is square, i.e., it has $m$ inputs and $m$ outputs, because the decoupling operation can only be applied to square systems. We will extend the theory to non-square systems in Section \ref{Sec: nonsquare system}, but for the ease of illustration, assume for now that $G(z)$ is a square system. Over the output, the constraints are imposed: $y_i(t) \in \mathbb{Y}_i, \forall t$, where $\mathbb{Y}_i$ are specified sets. Given the set-points $r_i$, the goal is to select each $u_i$ as close as possible to $r_i$ (to ensure that the tracking outputs, which are not shown in the figure, follow $r_i(t)$ as closely as possible) while ensuring that the output constraints are satisfied, i.e., $y_i(t) \in \mathbb{Y}_i, \forall t$. The DRG-tf method achieves these goals  as follows: first, system $G(z)$ is decoupled by finding a suitable filter, $F(z)$, that eliminates the coupling dynamics of $G(z)$. The resulting decoupled system is $W(z):= G(z) F(z)$, which is diagonal; that is, each output $y_i$ depends only  on the new input $v_i$. 
Second, we introduce a bank of $m$ decoupled SRGs, where the goal of the $i$-th SRG is to select $v_i$ as close as possible to $r_i'$ while ensuring $y_i \in \mathbb{Y}_i$. Each  SRG$_i$ (see Figure~\ref{fig: drg-tf system block}) uses only the states of the $i$-th decoupled subsystem. Finally, since we would like to ensure that $u_i = r_i$ when $r_i$ is constraint-admissible, we introduce the inverse of the filter, $F^{-1}(z)$, to cancel the effects of $F(z)$. Note that $F^{-1}(z)$ also ensures that $u_i$ and $r_i$ are close if $r_i$ is not constraint-admissible.

Similar to DRG-tf, DRG-ss is based on decoupling the input-output dynamics as shown in Figure~\ref{fig: system block ss}. The difference is that the system $G$ is decoupled by using state feedback, where the feedback matrices $\Phi$ and $\Gamma$ are properly chosen as will be discussed later in this paper. Second step is introducing $m$ decoupled SRGs, whose goal is the same as the SRGs in DRG-tf. Finally, to make sure that $u_i = r_i$ when $r_i$ is constraint-admissible, $x$ is fed back through  $\Gamma^{-1}(r-\Phi x)$.

\begin{figure}
\centering
\begin{tikzpicture}
    [L1Node/.style={rectangle,draw=black,minimum size=5mm},
    L2Node/.style={rectangle,draw=black, minimum size=13mm},
    L3Node/.style={rectangle,draw=black, minimum size=30mm}]
      \node[L2Node] (n3) at (1.4, 0){$F^{-1}(z)$};
      \node[L1Node] (n4) at (3.3, 0.5){$SRG_{1}$};
      \draw (3.2,0.1)node [color=black,font=\fontsize{10}{10}\selectfont]{\vdots};
      \draw (2.4,0.1)node [color=black,font=\fontsize{10}{10}\selectfont]{\vdots};
      \draw (4.2,0.1)node [color=black,font=\fontsize{10}{10}\selectfont]{\vdots};
      \draw (6.3,0.1)node [color=black,font=\fontsize{10}{10}\selectfont]{\vdots};
      \draw (8.3,0.1)node [color=black,font=\fontsize{10}{10}\selectfont]{\vdots};
      \draw (0.2,0.1)node [color=black,font=\fontsize{10}{10}\selectfont]{\vdots};
      \node[L1Node] (n4) at (3.3, -0.5){$SRG_{m}$};
      \node[L2Node] (n5) at (5.2, 0){$F(z)$};
      \node[L2Node] (n6) at (7.2, 0){$G(z)$};
	   \draw[dashed](4.4,1)-- (7.9,1);
      \draw[dashed](4.4,1)-- (4.4,-1);
      \draw[dashed](4.4,-1)-- (7.9,-1);
      \draw[dashed](7.9,-1)-- (7.9,1);
      \draw[->](0,0.5)--(0.65,0.5);
      \draw(6,1)--(6,1.2);
      \draw(6,1.2)--(3.2,1.2);
      \draw[->](3.2,1.2)--(3.2,0.8);
      \draw(6,-1)--(6,-1.2);
      \draw(6,-1.2)--(3.2,-1.2);
      \draw[->](3.2,-1.2)--(3.2,-0.8);

      \draw[->](0,-0.5)--(0.65,-0.5);
      \draw[->](2.1,0.5)--(2.6,0.5);
      \draw[->](2.1,-0.5)--(2.55,-0.5);
      \draw[->](3.95,0.5)--(4.5,0.5);
      \draw[->](3.99,-0.5)--(4.5,-0.5);       
      \draw[->](5.85,0.5)--(6.5,0.5);
      \draw[->](5.85,-0.5)--(6.5,-0.5);
      \draw[->](7.83,0.5)--(8.4,0.5);
      \draw[->](7.83,-0.5)--(8.4,-0.5);
      \draw (0,0.75)node [color=black,font=\fontsize{6}{6}\selectfont]{$r_{1}$};
      \draw (0,-0.75)node [color=black,font=\fontsize{6}{6}\selectfont]{$r_{m}$};
      \draw (2.4,0.75)node [color=black,font=\fontsize{6}{6}\selectfont]{$r_{1}'$};
      \draw (2.4,-0.75)node [color=black,font=\fontsize{6}{6}\selectfont]{$r_{m}'$};
        \draw (4.2,0.75)node [color=black,font=\fontsize{6}{6}\selectfont]{$v_{1}$};
      \draw (4.2,-0.75)node [color=black,font=\fontsize{6}{6}\selectfont]{$v_{m}$};
      \draw (6.3,0.75)node [color=black,font=\fontsize{6}{6}\selectfont]{$u_{1}$};
      \draw (6.3,-0.75)node [color=black,font=\fontsize{6}{6}\selectfont]{$u_{m}$};
      \draw (8.3,0.75)node [color=black,font=\fontsize{6}{6}\selectfont]{$y_{1}$};
      \draw (8.3,-0.75)node [color=black,font=\fontsize{6}{6}\selectfont]{$y_{m}$};
      \draw (6.5,-1.4)node [color=black,font=\fontsize{10}{10}\selectfont]{\hspace{1cm} $W(z)$};
      \draw (4.8,1.35)node [color=black,font=\fontsize{6}{6}\selectfont]{$x_{1}$};
        \draw (4.8,-1.35)node [color=black,font=\fontsize{6}{6}\selectfont]{$x_{m}$}; 
\end{tikzpicture}

\caption{DRG-tf block diagram.} \label{fig:Decoupled with RG}
 \label{fig: drg-tf system block}
\end{figure}
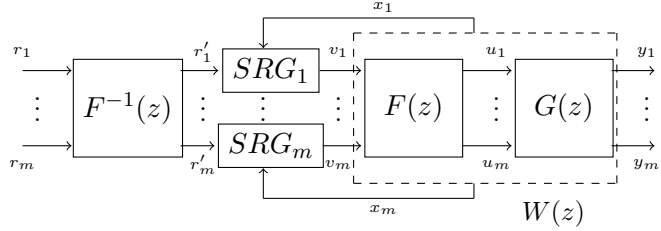

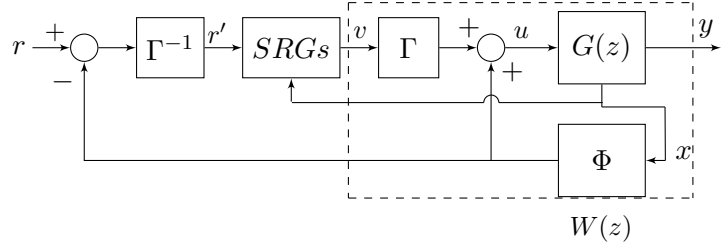
\begin{figure}
    \begin{center}
        \begin{tikzpicture}
        [node distance=1.5cm,>=latex']
        \node [input, name=input] {};
        \node [sum, right=0.5cm of input] (sum) {};
        \node [block,minimum size=0.8cm, right=0.5cm of sum] (Ginv) {$\Gamma^{-1}$};
        \node [block,minimum size=0.8cm, right=0.5 of Ginv] (RG) {$SRGs$};
        \node [block,minimum size=0.8cm, right =0.5 of RG] (Gmatrix) {$\Gamma$};
        \node [sum, right =0.5 of Gmatrix] (sum2) {};
        \node [block, right=0.7 of sum2] (system) {$G(z)$};
        \node [output, name=state1,below=0.3 of system] {};
        \node [output, name=state2,right=0.84 of state1] {};
        \node [output, right=1 of system] (output) {};
        \node [output, below of=sum2] (loop) {};
        \node [block, below of=system] (feedback) {$\Phi$};
        \node [output, name=state3,right=0.26 of feedback] {};
        \node [output, name=state4,above=0.6 of RG] {};
        \node [output, name=state5,below=0.31 of RG] {};
        \draw [draw,->] (input) -- node[above,pos=0.79,left,pos=0.1]{$r$} (sum);

        \draw [->] (sum) -- node {} (Ginv);
        \draw [->] (Ginv) -- node [above,pos=0.79]{}(RG);
        \draw [->] (RG) -- node [above,pos=0.79]{} (Gmatrix);
        \draw [->] (Gmatrix) -- node [above,pos=0.79]{$+$} (sum2);
        \draw [->] (sum2) -- node[above,pos=0.79] {} (system);
        \draw [->] (system) -- node [name=y][above,pos=0.8] {$y$}(output);
        \draw [->] (loop) -- node [right,pos=0.9] {$+$}(sum2);
        \draw [->] (feedback) -| node[left,pos=0.9] {$-$} 
        node [near end] {} (sum);
        \draw [-] (system) -- node [] {}(state1);
        \draw [-] (state1) -- node [] {}(state2);
         \draw [-] (state2) -- node [right,pos=0.85] {$x$}(state3);
        \draw [->] (state3) -- node [] {}(feedback);
        \draw [->] (state5) -- node [] {}(RG);
       
      \draw (6.5,0.2)node [color=black,font=\fontsize{10}{10}\selectfont]{$u$};

      \draw[dashed](4.2,0.6)-- (8.78,0.6);
      \draw[dashed](4.2,0.6)-- (4.2,-2);
      \draw[dashed](8.78,0.6)-- (8.78,-2);
      \draw[dashed](4.2,-2)-- (8.78,-2);
       
      \draw[-](7.58,-0.73)-- (6.2,-0.73);
      \draw[-](6,-0.73)-- (3.45,-0.73);
      \draw (6.2,-0.73) arc (0:180:0.1cm);
      \draw (7,-2.4)node
      [color=black,font=\fontsize{10}{10}\selectfont]{\hspace{1cm} $W(z)$};
       
      \draw (0.3,0.2)node
      [color=black,font=\fontsize{10}{10}\selectfont]{$+$};
       
      \draw (1.9,0.2)node
      [color=black,font=\fontsize{10}{10}\selectfont]{\hspace{1cm} $r'$};
       
      \draw (3.8,0.2)node
      [color=black,font=\fontsize{10}{10}\selectfont]{\hspace{1cm} $v$};

    \end{tikzpicture}
    \end{center}
    \caption{DRG-ss block diagram. $r$, $r'$, $v$, $u$, $y$ represent $[r_1, r_2,\ldots, r_m]^T$, $[r'_1, r'_2,\ldots, r'_m]^T$, $[v_1, v_2,\ldots, v_m]^T$, $[u_1, u_2,\ldots, u_m]^T$, and $[y_1, y_2,\ldots, y_m]^T$, respectively.}
    \label{fig: system block ss}
\end{figure}

Finally, we handle non-square systems by transforming them into square ones and applying the DRG theory explained above to the resulting square system. Detailed information will be provided in Section \ref{Sec: nonsquare system}. 

Because of the decoupling process, DRG-ss differs from DRG-tf in its analysis, implementation, and observer design. Furthermore, DRG-ss contains an additional feedback loop, which may compromise closed-loop stability.  Thus, in this paper, we present a detailed analysis of both methods, including stability, transient and steady-state properties, and observer design considerations. We also study the class of systems for which DRG performs well, and present an analysis of the robustness of DRG to unmeasured disturbances and parametric uncertainties. Note that a preliminary exposition of DRG-tf was presented in a conference version of this paper in \cite{DRG2018}. The current paper improves on \cite{DRG2018} by presenting a complete analysis of the transient and steady-state characteristics of DRG-tf, introducing and studying  DRG-ss, discussing observer design considerations, presenting the robustness analysis mentioned above, and introducing the  extension of  DRG-tf and DRG-ss to non-square systems.

The main contributions of this research are as follows:
\begin{itemize}
\item A computationally efficient constraint management technique for square MIMO systems (i.e.,  the DRG), which is a novel extension of the SRG. Two formulations of DRG (i.e., DRG-ss and DRG-tf), and their advantages and disadvantages, are studied.
\item Analysis of stability and performance of DRG in comparison with VRG. We show that  the  proposed approach is most suitable for a specific class of systems and illustrate this by examples.

\item Quantitative comparison of explicit and implicit  optimization techniques for VRG and DRG, where we show that DRG algorithm can run two orders of magnitude faster than VRG at every time step.

\item A novel extension of DRG to systems that are affected by unknown additive disturbances and parametric uncertainties.

\item An extension of DRG to non-square MIMO systems, which enhances the applicability of DRG.
\end{itemize}


\section{Preliminaries}\label{Sec: Preliminaries}

In this section, we  introduce the notations and norms that are used in this paper. Then, we review the  decoupling methods and reference governor schemes. 

The following notations are used in this paper. $\mathbb{Z}_+$ denotes the set of all non-negative integers. Let $V, U \subset \mathbb{R}^n$. Then, $V\sim U :=\{z\in \mathbb{R}^n: z+u \in V, \forall u \in U \}$ is the  Pontryagin-subtraction (P-subtraction) (\cite{Kolmanovsky_1998}).  The identity matrix with dimension $i \times i$ is denoted by $I_i$.
Given a discrete-time signal $u(t) =[u_1(t),u_2(t),\ldots,u_m(t)]^T$, the $L_2$ norm is defined as: $\|u\|_{L_2}^2 = \sum_{t=-\infty}^{\infty}u(t)^Tu(t)$, and its $L_\infty$ norm is represented as:  $\|u(t)\|_{L_\infty} = \sup_{t}(\max_{i}|u_i(t)|)$. For a system with transfer function $F(z)$ and impulse response $f(t)$,  the ${H_\infty}$ norm is defined as: 
\hbox{$\|F\|_{H_\infty} = \max_{w}\Bar{\sigma}(F(e^{jw}))$}, where $\Bar{\sigma}$ represents the maximum singular value, and the $L_1$ norm is defined as: $\|f(t)\|_{L_1} = \max_{i}\sum_{j=1}^{m}\sum_{\tau=0}^{\infty}|f_{ij}(\tau)|$, where $f_{ij}$ is the  $ij$-th element of $f$, and $m$ is the number of columns of $f$. We denote the condition number of a matrix (defined by the ratio of the maximum to the minimum singular values) by $\gamma$. A zero matrix with dimension  $i \times j$ is denoted as $0_{i,j}$.

\subsection{Review Of Decoupling Methods}\label{Section: Review Decoupled methods}

In this section, we review two decoupling methods, one based on transfer functions (\cite{Skogestad_2007}) and the other based on state space (\cite{falb1967decoupling}).

\subsubsection{Decoupling Method Based on Transfer Functions}\label{Sec. ef decoupling}

Consider the square coupled  system $G(z)$ shown in Figure~\ref{fig: drg-tf system block} and defined as: 
\begin{equation}\label{eq:Gz}
\renewcommand*{\arraystretch}{.5}
\begin{bmatrix}
Y_1(z)\\
\vdots\\
Y_m(z) \end{bmatrix}
=\underbrace{\begin{bmatrix}
      G_{11}(z) & \ldots & G_{1m}(z)\\
      \vdots & \ddots & \vdots \\
      G_{m1}(z) & \ldots & G_{mm}(z)
     \end{bmatrix}}_{G(z)} \begin{bmatrix}
U_1(z)\\
\vdots\\
U_m(z) \end{bmatrix}
 \end{equation}
where $Y_i$ and  $U_i$ are the $\mathcal{Z}$-transforms of $y_i$ and $u_i$, respectively. The system $G(z)$ consists of  diagonal subsystems with dynamics $G_{ii}(z)$ and off-diagonal (interaction) subsystems with dynamics  $G_{ij}(z), i\neq j$. 
A decoupled system is perfectly diagonal (i.e., each output depends on only one input). As shown in Figure~\ref{fig: drg-tf system block}, we decouple the system  by adding a filter, $F(z)$, before $G(z)$, so that the product $G(z) F(z)$ yields a diagonal transfer function matrix  \(W(z):=G(z) F(z)\) (\cite{Skogestad_2007}). By doing so, each output $Y_i$  depends only on the new input $V_i$ through: $Y_i(z) = W_{ii}(z) V_i(z)$, where $W_{ii}(z)$ is the $i$-th diagonal elements of $W(z)$ and $V_i(z)$ is the $\mathcal{Z}$-transform of $v_i$. 

In this paper, we study two structures for $W(z)$, which lead to the following two  decoupling methods:
\begin{itemize}
\item Diagonal Method:
We find $F(z)$ such that $W(z)=\mathrm{diag}(G_{11},G_{22},\ldots,G_{mm})$. The filter and the inverse filter are defined  as:
\begin{equation}\label{eq:diagonal method}
F(z)=G^{-1}(z) W(z), \;\;\; F^{-1}(z)=W^{-1}(z)  G(z)
\end{equation}

\item Identity Method:
We find $F(z)$ such that $W(z)$ equals the identity matrix. The filter and the inverse filter are defined  as:
\begin{equation}\label{eq:intentical method}
F(z)=G^{-1}(z), \;\;\;\; F^{-1}(z)= G(z)
\end{equation}
\end{itemize}

Notice that in both methods, the elements of  either $F(z)$ or $F^{-1}(z)$ (or both) may be improper transfer functions because of $G^{-1}(z)$ and $W^{-1}(z)$. If this is the case, they cannot be implemented in the DRG scheme of Figure~\ref{fig: drg-tf system block}. In order to make them proper, we multiply $F(z)$ and $F^{-1}(z)$ by time-delays of the form $\frac{1}{z^{\beta}}$, where $\beta$ refers to how much time delay should be added to make the transfer functions proper. Note that if delays are added to either $F$ or $F^{-1}$, the system response will be delayed under the DRG scheme, even if no constraint violation is likely. This is a caveat of  the DRG approach; however, if the sample time is small enough, the introduced delay would be negligible. Also note that $G^{-1}(z)$ might introduce unstable poles to $F(z)$ or $F^{-1}(z)$, which will cause the system to become unstable. Further assumptions are introduced later in the paper to avoid such situations.

\begin{remark}
In the above discussion, the matrix $W(z)$ is assumed to be diagonal, which means that every $y_i$ depends only on $v_i$. This, however, is only one possible structure for $W(z)$. It is also possible to decouple the system by having each $y_i$ depend on one $v_j$, $j\neq i$. In this case, the structure of $W(z)$ will be such that every row will have only one non-zero element. Similarly, each column will have only one non-zero element. The DRG scheme presented in this paper can be used with this structure of $W$. However, for the sake of simplicity, in the rest of this paper, we will assume that $W$ is diagonal. 
\end{remark}

\subsubsection{Decoupling Method based on State Feedback}\label{Sec: feedback decoupling}

In this section, we describe input/output decoupling via state-feedback, as presented in 
\cite{falb1967decoupling, lloyd1970decoupling}. Consider a discrete-time coupled system, $G$ (see Figure~\ref{fig: system block ss}), given in state-space form by:
\begin{equation}\label{eq:state-space model}
\begin{aligned}
&x(t+1) =Ax(t)+Bu(t)\\
&y(t)=Cx(t)+Du(t)
\end{aligned}
\end{equation}
where $x\in \mathbb{R}^n$ is the state vector, $u\in \mathbb{R}^m$ is the input, and $y\in \mathbb{R}^m$ is the output vector. Note that the number of inputs is equal to the number of outputs.

In the remainder of this discussion, we assume no direct feed through between $u$ and $y$ (i.e., $D=0$) as required by \cite{falb1967decoupling, lloyd1970decoupling}. Note that the case where $D\neq 0$ can be handled as well (e.g., see  \cite{silverman1970decoupling}), but for the sake of simplicity, here we will only present the case where $D=0$. 

The substitution of $u=\Phi x+\Gamma v$, where $\Phi$ is an $m\times n$ matrix and $\Gamma$ is an $m \times m$ matrix, into \eqref{eq:state-space model} results in:
\begin{equation}\label{eq: Astar and Bstar}
x(t+1) = \underbrace{(A+B\Phi)}_{\bar{A}}x(t)+\underbrace{B\Gamma}_{\bar{B}} v(t), \;\;
y(t) = Cx(t)
\end{equation}
Let $d_1,d_2,\ldots,d_m$ be defined by:
\[ d_i= \text{min}\{j: C_iA^jB \neq 0, j=0,1,\ldots,n-1\} \]
where $C_i$ denotes the $i$-th row of $C$. If $C_iA^jB = 0$ for all $j=0,1,\ldots,n-1$, then we set  $d_i=n-1$. Let $A^*{\in \mathbb R}^{m \times n}$ and $B^*{\in \mathbb R}^{m \times m}$ be defined by:
\begin{equation}\label{eq:B_star}
\renewcommand*{\arraystretch}{.6}
A^* = \begin{bmatrix}
      C_1A^{d_1+1}\\
      \vdots\\
      C_mA^{d_m+1}\\
\end{bmatrix}, B^* = \begin{bmatrix}
      C_1A^{d_1}B\\
      \vdots\\
      C_mA^{d_m}B\\
\end{bmatrix} 
\end{equation}

It is shown in \cite{falb1967decoupling} that there exist a pair of matrices $\Phi$ and $\Gamma$ that decouple the system from $v$ to $y$ if and only if $B^*$ is nonsingular. 

Below, we study two structures for $\Phi$ and $\Gamma$, which lead to the following two  decoupling methods:

\begin{itemize}
\item Identity method:
The pair
\begin{equation}\label{eq:pair1}
\Phi = -B^{*-1}A^*, \:\:\:\:\:\: \Gamma= B^{*-1}
\end{equation}
leads to $y_i(t+d_i+1)=v_i(t)$, which means that the $i$-th output depends only on the $i$-th input with one or more time delays.

\item Pole-assignment method:
We can decouple the system while simultaneously assigning the poles of the decoupled system by using the following choice of $\Gamma$ and $\Phi$:
\begin{equation}\label{eq:pair2}
\Phi = B^{*-1}\left[\sum_{k=0}^{\delta}M_kCA^k-A^*\right], \;\;
\Gamma = B^{*-1}
\end{equation}
where $\delta=\max d_i$ and $M_k$ are $m \times m$ diagonal matrices that are designed to assign the poles at specific locations. For more details, please see \cite{falb1967decoupling}. Note that not all of the eigenvalues of $\bar{A}$ can be arbitrarily assigned. However, it is shown in \cite{falb1967decoupling} that if $m+\sum_{i=1}^{m}d_i=n$, all the poles of the decoupled system can be assigned.
\end{itemize}

\subsection{Review Of Reference Governors}\label{Section:review RG}

This section reviews the  scalar and vector reference governors as presented in \cite{Kolmanovsky_2014,GARONE2017306}.  Consider a discrete-time square linear system described by the state-space model in \eqref{eq:state-space model}.
Suppose $y(t)\in \mathbb{R}^m$ is the constrained output vector, over which the following constraints are imposed: $y_i(t) \in \mathbb{Y}_i, i=1,\ldots,m,  \forall t\in \mathbb{Z_+}$, where $\mathbb{Y}_i\subset \mathbb{R}$ are specified constraint sets. The constraints can also be expressed in vector form as:  $y(t) \in \mathbb{Y},$ where $\mathbb{Y}$ is given by the Cartesian product $\mathbb{Y} = \mathbb{Y}_1 \times \cdots \times \mathbb{Y}_m$. A review of SRG and VRG is provided next.

\subsubsection{Scalar Reference Governor (SRG)}\label{sec:srg}

SRG computes the input $u(t)$ in \eqref{eq:state-space model} as a convex combination of the previous input $u(t-1)$ and the  current reference $r(t)$, i.e., 
\begin{equation}\label{eq:kapa}
u(t)=u(t-1)+\kappa(r(t)-u(t-1))
\end{equation}
where $\kappa$ is the solution of the following linear program:
\begin{equation}\label{eq: LP for RG}
\begin{aligned}
&\underset{\kappa\in [0,1]}{\text{maximize}}
& & \mathrm{\kappa} \\
& \hspace{10pt} \text{s.t.}
& & u(t)=u(t-1)+\kappa(r(t)-u(t-1))\\
&&&(x(t),u(t)) \in O_\infty
\end{aligned}
\end{equation}
where $O_\infty$ is the Maximal Admissible Set (MAS) discussed below. In the above optimization problem, $x(t)$, $r(t)$, and $u(t-1)$ are known parameters, and $\kappa$ is the optimization variable. Note that $\kappa=0$ means that in order to keep the system safe, $u(t)=u(t-1)$, and $\kappa=1$ means that no violation is predicted and, therefore, $u(t)=r(t)$. This RG formulation ensures closed-loop stability and recursive feasibility. Note that if constraint violation is predicted for any output, all inputs will be affected equally because a single $\kappa$ is used. Thus, the response of SRG may be conservative for MIMO systems.

MAS is the set of all safe initial conditions and inputs, defined as: 
$$
O_\infty := \{(x_0,u_0): x(0)=x_0, u(t)=u_0, y(t)\in \mathbb{Y}, \forall t\geq0\}
$$
where it is assumed that $u(t) = u_0$ is held constant for all time. Computation of MAS is possible, as $y(t)$ can be expressed explicitly as a function of $x(0)=x_0$ and $u_0$:
$y(t)=CA^tx_0+(C(I-A)^{-1}(I-A^t)B+D)u_0$. 
MAS can be computed using the above, and can be shown to be a polytope of the form:
\begin{equation}\label{eq:OinfForm}
O_\infty = \{(x_0,u_0): H_x x_0 + H_u u_0 \leq h\}
\end{equation}
Conditions for $O_\infty$ to be  finitely determined (i.e., matrices $H_x, H_u, h$ to be finite dimensional) are  discussed in \cite{Ilya_1995,Gilbert_1995}. Basically, to ensure that $O_\infty$ is finitely determined, the steady-state constraint is first tightened, resulting in the steady-state admissible set, $P_{ss}$:
\begin{equation}\label{eq: O_inf steady state}
 P_{ss}:=\{(x_0,u_0): G_0u_0 \in \mathbb{Y}_{ss} \}
\end{equation}
where $G_0$ is the DC  gain of system \eqref{eq:state-space model}, and ${\mathbb{Y}_{ss}:=(1-\epsilon)\mathbb{Y}}$ for some small positive $\epsilon$. The intersection of  $P_{ss}$ with $O_\infty$ (i.e., adding the inequality in \eqref{eq: O_inf steady state} to \eqref{eq:OinfForm}) leads to a finitely determined inner approximation of $O_\infty$. In the sequel, with some abuse of notation, we assume that $O_\infty$ includes the tightened steady-state constraint and is, hence,  finitely determined.

\subsubsection{Vector Reference Governor (VRG)}

VRG extends the capabilities of  SRG and uses diagonal matrix $K$  instead of scalar $\kappa$.  Eq. \eqref{eq:kapa} is reformulated as:

\begin{equation}\label{output of RG}
  u(t)=u(t-1)+K(r(t)-u(t-1))  
\end{equation}

where $K=\mathrm{diag}(\kappa_i)$. The values of  $\kappa_i, i=1,...,m$, are chosen by solving a Quadratic Program (QP):
\[\begin{aligned}
&\underset{\kappa_i \in [0,1]}{\text{minimize}}
& & \|u(t)-r(t)\| \\
& \hspace{10pt}\text{s.t.}
& & u(t)=u(t-1)+K(r(t)-u(t-1))\\
&&&(x(t),u(t)) \in O_\infty
\end{aligned}\]
Note that for VRG, $O_\infty \subset \mathbb{R}^{n+m}$ can be computed in the same way as explained in Subsection \ref{sec:srg}. Because of the increased number of optimization variables and the QP formulation, VRG is more computationally demanding than SRG.

\subsubsection{Maximal Admissible Sets (MAS) for systems with disturbances}\label{sec:MAS with disturbance}

In this section, we review  the concept of robust MAS for systems affected by additive disturbances:
\begin{equation}\label{eq:system with disturbance}
\begin{aligned}
&x(t+1)=Ax(t)+Bu(t)+B_{w}w(t)\\
&y(t)=Cx(t)+Du(t)+D_{w}w(t) \in \mathbb{Y}
\end{aligned}
\end{equation}
where $\mathbb{Y}$, as before, is the constraint set. The  disturbance input satisfies $w \in \mathbb{W}$, where $\mathbb{W}\subset \mathbb{R}^d$ is a compact polytope with the origin in its interior. 
Works that have explored unknown disturbances for RG schemes can be found in \cite{Osorio_2018,Kolmanovsky_1998,Gilbert_1999,Osorio_2019}. 

 In order to define the robust MAS for system \eqref{eq:system with disturbance}, we write $y(t)$ as a function of the initial state, $x_0$, the constant input, $u(t) = u_0$, and the disturbances:
\begin{equation}\label{eq: model recursion2}
\begin{aligned}
y(t)&=CA^tx_{0}+(C(I-A)^{-1}(I-A^t)B+D )u_{0}\\
& + C\sum_{j=0}^{t-1}A^{t-j-1}B_{w} w(j) + D_{w} w(t)
\end{aligned}
\end{equation}
We now define the sets $\mathbb{Y}_{t}$ using the following recursion:
\begin{equation}\label{eq: Y_0 and Y_t}
\mathbb{Y}_{0}=\mathbb{Y}\sim D_w \mathbb{W},\quad
\mathbb{Y}_{{t+1}}=\mathbb{Y}_{t}\sim CA^{t}B_w \mathbb{W}
\end{equation}
P-subtraction allows us to  rewrite the requirement ${y(t) \in \mathbb{Y}}$, $\forall w(j)\in \mathbb{W}, j = 0,\ldots,t$ as:
$$
CA^t x_{0}+ (C(I-A)^{-1}(I-A^t)B+D)u_0 \in \mathbb{Y}_{t}
$$

Finally, we define the robust MAS as:
\begin{equation}\label{eq:O_t definition}
\begin{aligned}
O&_\infty := 
\{(x_{0},u_{0})\in \mathbb{R}^{n+m}: G_{0} u_{0} \in \bar{\mathbb{Y}}, \\
&CA^t x_{0}+ (C(I-A)^{-1}(I-A^t)B+D)u_{0} \in \mathbb{Y}_{t}\}
\end{aligned}
\end{equation}
where $G_{0}$ is the DC gain of  (\ref{eq:system with disturbance}) from $u$ to $y$, and $\bar{\mathbb{Y}}:=(1-\epsilon)\mathbb{Y}_{t}$ for some $0 < \epsilon \ll 1$ and large $t$. Note that $\bar{\mathbb{Y}}$ is introduced to ensure finite-determinism of $O_{\infty}$ (similar to Section \ref{sec:srg}).

\section{Decoupled Reference Governor based on Transfer Function Decoupling: DRG-tf}\label{sec: drg-tf}

As mentioned in the Introduction (see Figure~\ref{fig: drg-tf system block}), DRG-tf is based on decoupling the system using the method described in Section \ref{Sec. ef decoupling} to obtain a completely diagonal system $W(z)$, where $W(z):=\mathrm{diag}(W_{11}(z),\ldots,W_{mm}(z))$, followed by implementing $m$ independent scalar reference governors for the resulting decoupled subsystems, and coupling the dynamics using $F(z)^{-1}$ to cancel the effects of $F(z)$. 
Because the SRGs are inherently nonlinear elements, one challenge with DRG-tf is  quantifying the tracking performance of the system. State estimation is another challenge. More specifically, how can the states of the  decoupled subsystems be obtained and fed back to the SRGs? 
In this section, we elaborate on DRG-tf with a special focus on the above challenges.

The following assumptions are made in this section:

\begin{assumption}\label{A: system stable}
System $G(z)$  in Figure~\ref{fig:Decoupled with RG} 
reflects the combined closed-loop dynamics of the plant with a stabilizing controller. Consequently, $G(z)$ is asymptotically stable (i.e., $|\lambda_i(A)|<1, i=1,...,n$). Furthermore, we assume that all diagonal subsystems of the decoupled system $W(z)$ are also asymptotically stable.
\end{assumption}

\begin{assumption}\label{A: NMP}
 $G(z)$ in Figure~\ref{fig: drg-tf system block} is invertible and has a stable inverse. 
\end{assumption}

\begin{assumption}\label{A: intervals}
The constraint sets $\mathbb{Y}_i$ are closed intervals of the real line containing the origin in their interiors. This is in agreement with the assumptions commonly made in the literature of reference governors. We thus assume that $\mathbb{Y}_i = \{y_i: \underline{s} \leq y_i \leq \bar{s}$\}. 
\end{assumption}

Consider the system in Figure~\ref{fig:Decoupled with RG} with  $G(z)$ given in  \eqref{eq:Gz}.
To design the SRGs, we compute the maximal admissible set (MAS) for each $W_{ii}$, denoted by $O_{\infty}^{W_{ii}}$. 
To obtain these sets, we find a minimal state-space realization of each subsystem $W_{ii}$, and compute its MAS as:
\begin{equation}\label{eq:O_inf for DRG}
\begin{aligned}
O_{\infty}^{W_{ii}}:=&\{(x_{i_0},v_{i_0})\in \mathbb{R}^{n_i+1}: x_i(0)=x_{i_0},\\
& v_i(t)=v_{i_0}, y_i(t)\in \mathbb{Y}_i, \forall t \in \mathbb{Z}_+ \}
\end{aligned}
\end{equation}
where $x_i$ and $n_i$ are the state and the order of $W_{ii}$, respectively. If the states of $G$ are unknown, an observer can be designed, which will be explained later. 

The DRG-tf formulation is based on the sets $O_{\infty}^{W_{ii}}$. Specifically, the inputs $v_i$ are defined, similar to \eqref{eq:kapa}, by:
\begin{equation}\label{eq:diagonalRG}
v_i(t)=v_i(t-1)+\kappa_i(r'_i(t)-v_i(t-1))
\end{equation}
where $\kappa_i$ are computed by $m$ independent linear programs:
\begin{equation}\label{eq:kappa for DRG}
\begin{aligned}
&\underset{\kappa_i\in [0,1]}{\text{maximize}}
& & \mathrm{\kappa_i} \\
& \hspace{10pt} \text{s.t.}
& & v_i(t)=v_i(t-1)+\kappa_i(r'_i(t)-v_i(t-1))\\
&&&(x_i(t),v_i(t)) \in O_{\infty}^{W_{ii}}
\end{aligned}
\end{equation}

\begin{remark}
Note that, since $F(z)$ and $F^{-1}(z)$ are both assumed to be stable, the DRG formulation above inherits the stability and recursive feasibility properties of standard SRG theory. Specifically, for a constant signal ${r(t)=r}$, $r'(t)$ converges (because of stability of $F^{-1}$), which implies that $v(t)$ converges (because of stability  of SRGs). Thus, the system of Figure~\ref{fig: drg-tf system block} is guaranteed to be stable.
\end{remark}

Below, we specialize the DRG-tf formulation to the two decoupling methods presented in Section \ref{Sec. ef decoupling}, namely the diagonal and the identity methods. We address the subtleties associated with both methods, including observer design and the structure of the maximal admissible sets. We also illustrate, using examples, that DRG-tf is most effective for systems with small condition numbers and singular values. We will provide a theoretical basis for this observation in Section \ref{sec:analysisDRGtf}. Finally, we show, using examples and analysis, that the DRG-tf with identity method may not perform well for certain systems even if the plant model is known precisely.

\subsection{Diagonal method}\label{sec:diag}

Recall that decoupling using the diagonal method leads to the decoupled system  $W(z):=\mathrm{diag}(G_{11},\ldots,G_{mm})$. Thus, we assume that $O_{\infty}^{W_{ii}}$ are created using minimal realizations of $G_{ii}$ (i.e., the diagonal subsystems of the original system). In this section, we first present an example to highlight the key attributes of the DRG-tf with diagonal method. For this example, we assume that the states of all $G_{ii}$ are measured and are available for feedback to the SRGs. This assumption will be relaxed in the next subsection, which discusses the issue of observer design.

\subsubsection{Motivating example} 
Consider the system $G(z)$ in \eqref{eq:Gz} given by: \begin{equation}\label{eq: illustrative example1}
G(z)=\begin{bmatrix}
      \frac{0.9}{(z-0.2)^2}&\frac{q}{3z+1}\\
      \frac{3}{(2z-1)^2}&\frac{0.4}{z-0.6}\\
     \end{bmatrix}
\end{equation}
 and the constraints  defined by $-1.2 \leq y_1 \leq 1.2$ and  $-3.9 \leq y_2 \leq 3.9$. The parameter $q$ will be selected later. 
 
 \begin{figure}
\begin{center}
\includegraphics[scale=0.45]{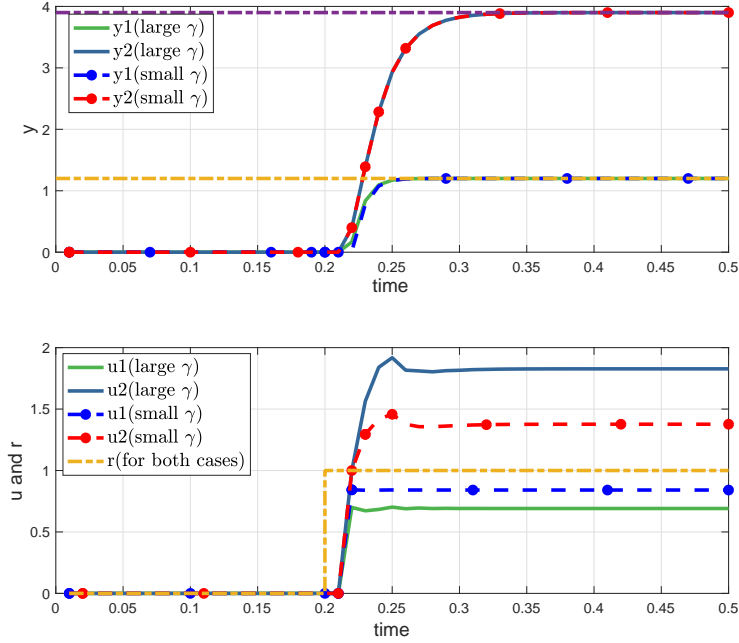}
\caption{Comparison of DRG-tf for systems with small and large condition numbers ($\gamma$). Top plot is the output (constraints shown by dashed lines) and the bottom plot is the reference $r(t)$ and the plant input $u(t)$.}
\label{fig: example1}
\end{center}
\end{figure}
\begin{figure}
\centering	
\includegraphics[scale=.65]{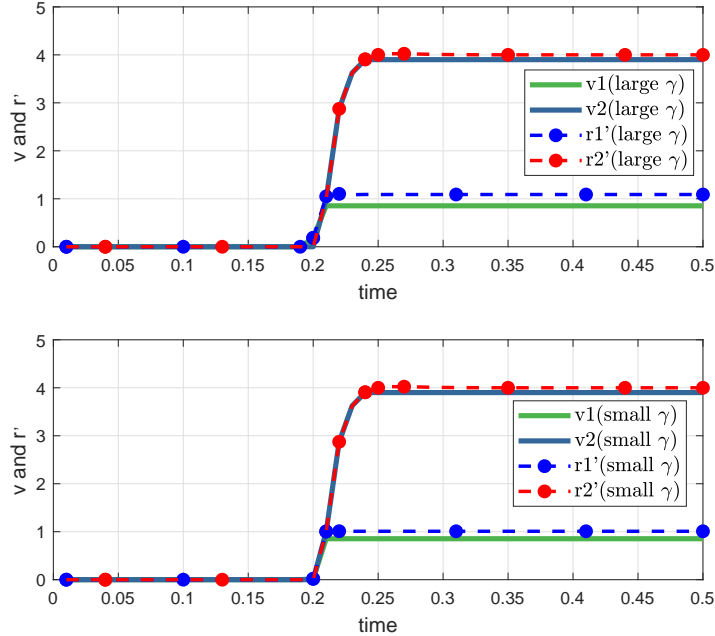}
\caption{Comparison of $r'(t)$ and $v(t)$ in DRG-tf with large (top plot) and small (bottom plot) condition number systems.}
\label{fig: example1randv}
\end{figure}

Next, we use \eqref{eq:diagonal method} to find $F(z)$. Noticing that in this example, we encounter the situation that both $F(z)$ and $F^{-1}(z)$ are not proper, we multiply them by $\frac{1}{z}$. Finally, we obtain the decoupled system: $W(z)=\frac{1}{z} \mathrm{diag}(\frac{0.9}{(z-0.2)^2}, \frac{0.4}{z-0.6})$.  In the discussion below, we denote the DC-gain matrix of $F(z)$ by $F_0$.

As mentioned in the Introduction, a requirement for DRG is that the signals $u(t)$ and $r(t)$ should be  as close as possible to ensure that the tracking performance of the system does not degrade significantly as compared with  VRG. As it turns out, this will be the case if the maximum singular value of  $F_0$, i.e., $\Bar{\sigma}$, is small. On the other hand, this will not be the case if the minimum singular value of  $F_0$, $\underline{\sigma}$, is large. We will analytically prove these statements in Section \ref{sec:analysisDRGtf}, but here we illustrate them via our example. For this, we consider two different $q$'s in \eqref{eq: illustrative example1}: $q=0.5$ and $q=0.05$.  If $q=0.5$, then the maximum and minimum singular values of $F_0$ are $\Bar{\sigma}=4.51$ and $\underline{\sigma}=0.30$, and its condition number is $\gamma=14.94$. If $q=0.05$, then $\Bar{\sigma}=3.39$, $\underline{\sigma}=0.30$, and $\gamma=11.21$. The second case has smaller condition number and maximum  singular value compared to the first one.

We proceed to design the DRG-tf  based on $W(z)$. In this case, we obtain $O_{\infty}^{W_{11}}$ and $O_{\infty}^{W_{22}}$ based on \eqref{eq:O_inf for DRG}. We simulate the response of this system to a step of size 1 in both $r_1$ and $r_2$. The simulation results  for both $q=0.5$ and $q=0.05$ are depicted in Figure~\ref{fig: example1}.

\begin{figure}
\centering	
\includegraphics[scale=0.4]{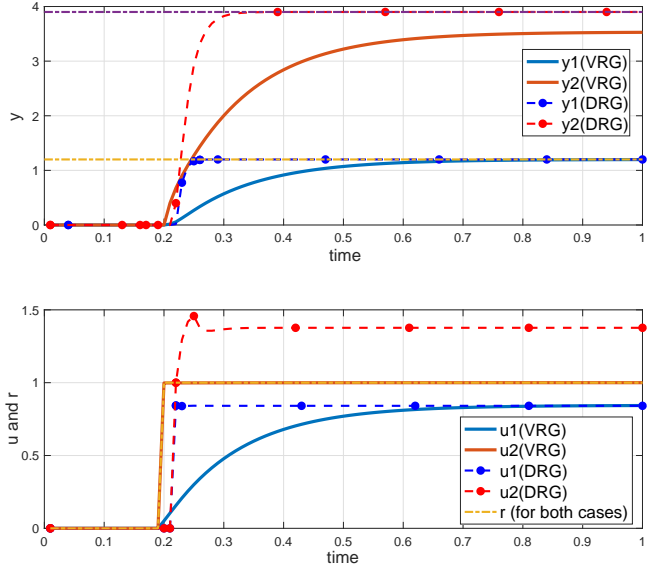}
\caption{Comparison of VRG and  DRG-tf for the system shown in \eqref{eq: illustrative example1} with $q=0.05$. Top plot is the output. The dashed lines are the output constraints. Bottom plot is $u$ compared with $r$. The dashed yellow line is the reference.}
\label{fig: vrg compared drg}
\end{figure}

As can be seen, the outputs in both cases satisfy the constraints, as required. However,  $u(t)$ is  closer to $r(t)$ for the system with smaller condition number and singular value (i.e., for $q=0.05$), which indicates better tracking performance. 

Furthermore, it can be seen from Figure~\ref{fig: example1randv} that $v(t)$ is always below $r'(t)$, which is expected since in SRG theory, the output of SRG is always bounded above (or below) by its input (in this case,  $r'$). However, note from Figure~\ref{fig: example1} that  $u(t)$ may be above $r(t)$, which is a situation that does not arise in SRG or VRG applications. The reason can be explained as follows 
(see Figure~\ref{fig: drg-tf system block}): at steady state, $u$ converges to $F_0 v$ and $r'$ converges to $F_0^{-1} r$. Thus, $r-u$ converges to $F_0 (r'-v)$, which indicates that, even if $r_i'>v_i$, $u_i$ may be above or below $r_i$ 
depending on $F_0$. Note that $u_i$ above $r_i$ may or may not be acceptable depending on the specific application. An example where this situation is acceptable is a distillation process \cite{Skogestad_2007},  because the constrained outputs (i.e., product compositions) determine the efficiency of the process. On the other hand, an example where this situation is not acceptable is in aerospace applications like controlling a drone, since the roll, pitch, and yaw angles cannot exceed their commands.

Figure~\ref{fig: vrg compared drg} shows a comparison between VRG and DRG-tf for $q=0.05$. There is a time delay in the  response of DRG-tf that is caused by the delay added to $F$ and $F^{-1}$ to make them proper. Note that $u(t)$ is below $r(t)$ for the VRG but not for DRG-tf, as explained above. More interestingly,  the rise time for DRG-tf is much faster than that of VRG. This is because, for this example, the interacting dynamics are slow and dominant, which causes the VRG to generate slow inputs. The DRG-tf, on the other hand, operates on the decoupled system where these slow dynamics have been canceled. This shows that, in addition to computational advantages, the DRG-tf may also have performance advantages compared to VRG.

To investigate the above observation more thoroughly, a comparison between the volumes for the MAS's of DRG-tf (i.e., volumes of $O_{\infty}^{W_{11}}$ and $O_{\infty}^{W_{22}}$) and VRG (i.e., volume of $O_\infty$) is as follows for $q=0.05$: $\mathrm{vol}(O_{\infty}^{W_{11}})=1.698$, $\mathrm{vol}(O_{\infty}^{W_{22}})=7.761$, and $\mathrm{vol}(O_{\infty})=7.761$ (volumes are computed using the MPT toolbox that is introduced in \cite{Herceg_2013}). Clearly, the sum of the volumes for DRG-tf exceeds the volume of the MAS for VRG, which is in agreement with the observations of Figure~\ref{fig: vrg compared drg} regarding a less conservative response from DRG-tf in comparison to VRG. A deeper analysis of the geometric properties of the MAS's is outside of the scope of this paper; however, it may be considered as an interesting topic for future work.

\subsubsection{Observer design} \label{sec:DRG-tf observer}

In this section, we consider the case where the states of $W_{ii}$, or equivalently $G_{ii}$, are not measured. Indeed, an observer will be required to estimate the states. One option is to use an open loop observer for each $W_{ii}$. To explain, let $(A_{ii},B_{ii},C_{ii},D_{ii})$ be a minimal realization of $W_{ii}$. An open loop observer can  be designed by computing the state estimate recursively:
\begin{equation}\label{eq:open loop observer}
\hat{x}_i(t+1) =A_{ii}\hat{x}_i(t)+B_{ii}v_i(t)
\end{equation}
where $\hat{x}_i$ is the estimate of the state $x_i$. In real-time, the SRGs in the DRG-tf formulation use $\hat{x}_i$ instead of $x_i$. Note that the open loop observer works well only when the system model and the initial conditions are both accurately known, which is not always the case.

To improve upon the open loop observer, feedback can be implemented from the measured output, as is done in standard observer design. We consider two observer design strategies below. The first assumes that all $y_i$ are measured, which leads to $m$ decoupled observers, and the second assumes that some $y_i$ are not measured, necessitating a centralized observer. Both strategies lead to subtleties for DRG-tf that we highlight in this section. 

\textbf{Decoupled observers:}
First suppose that all $y_i$ are available for measurement. In this case, we can design $m$ decoupled Luenberger observers as follows:
\begin{equation}\label{eq:observerClosed}
\hat{x}_i(t+1) =A_{ii}\hat{x}_i(t)+B_{ii}v_i(t) + L_i (y_i(t) - C_{ii} \hat{x}_i(t) - D_{ii} v_i(t))
\end{equation}
where $L_i$ is designed to assign the eigenvalues of $A_{ii}-L_iC_{ii}$ in the unit circle. Note that for the DRG-tf implementation, the state that feeds back to SRG$_i$ is  $\hat{x}_{i}$.
 
 A challenge with the above observer is that of selecting the initial conditions for each $\hat{x}_i$. Indeed, if the observers are not initialized properly, the DRG-tf scheme may not be able to enforce the constraints. We provide a solution to this problem below, for the case where the initial condition of $G(z)$, denoted by $x_0$, is known precisely. We will treat the case of unknown $x_0$ later.

Our solution is to modify the input to $G(z)$ in Figure~\ref{fig: drg-tf system block} to explicitly cancel the effects of $x_0$. To see how this can be done, note that the output of $G(z)$ with initial condition $x_0$ can be written as:
\( y(t)=CA^tx_0+(C(I-A^t)(I-A)^{-1}B+D)u_0  \), where $A$, $B$, $C$, and $D$ are the state space matrices of $G$. Denote by $M(z)$ the $\mathcal{Z}$-transform of $CA^t$ for the sake of simplicity of notation.  Note that $M(z)$ represents the initial condition response of the system. In order to get $Y(z) = W(z) V(z)$, where $W$ is a desired diagonal matrix as before, we define $U(z)$ as:
\begin{equation} \label{eq: augment initial condition}
U(z) = F(z) V(z)  - G^{-1}(z)M(z) x_0
\end{equation}
where $F(z)=G(z)^{-1} W(z)$ as before (compare \eqref{eq: augment initial condition} with $U(z) = F(z) V(z)$ in Figure~\ref{fig: drg-tf system block}). This will effectively cancel the initial conditions and result in a completely decoupled system.  The observers given in \eqref{eq:observerClosed} and the SRGs can now be applied as before. Note that the inverse filter $F^{-1}(z)$ in Figure~\ref{fig: drg-tf system block} need not be altered.
 
 \textbf{Centralized observer}:
 Now consider the more interesting case, where either some $y_i$ are not measured, or outputs other than $y_i$ are measured. Since the dynamics from $v$ to $y$ are still required to be decoupled, $m$ decoupled SRGs can still be used in the DRG-tf formulation. However,  we can not design $m$ decoupled observers for each $W_{ii}$ as we did before (since independent measurements are not available), and must instead design one centralized observer for $W$. This, in turn, implies that the SRGs must use a MAS different  from \eqref{eq:O_inf for DRG}. To elaborate on these ideas, let $y(t)$, as before, denote the constrained output vector, and let $y_m$ denote the measured output vector. Let $(A,B,C,D)$, $(A,B,C_m,D_m)$, and $(A_f,B_f,C_f,D_f)$ be realizations of $G(z)$ from $u$ to $y$,   $G(z)$ from $u$ to $y_m$, and $F(z)$, respectively. The states of $F(z)$, $x_f$, are known at the time of implementation so they do not need to be estimated. To estimate the states of $G(z)$, $x$, an observer is designed using feedback on the measurements $y_m$:
\begin{equation}\label{eq:observerClosed2}
\hat{x}(t+1) =A\hat{x}(t)+B u (t) + L (y_m(t) - C_m \hat{x}(t) - D_m u(t))
\end{equation}
 Using the above,  the states of the entire system, i.e., $x_w=(x_f, x)$, can be estimated by  $\hat{x}_w=(x_f, \hat{x})$. Note that initialization of this observer is simple if the initial condition of $G(z)$, i.e., $x_0$ is known: in this case, the initial condition of the observer is set to $(0,x_0)$. 
 
 Recall that to construct $O_{\infty}^{W_{ii}}$, the state-space model of the $i$-th diagonal subsystem of $W$, $W_{ii}$, is required.
 However, the states of each individual $W_{ii}$ is not directly available, which is why the SRGs can no longer use the $O_{\infty}^{W_{ii}}$ sets as described in \eqref{eq:O_inf for DRG}. To remedy this, we use the following realization of $W$, which is the augmented dynamics of $F(z)$ and $G(z)$: 
 \begin{equation}\label{eq: realizaiton of W}
 \renewcommand*{\arraystretch}{.5}
\begin{aligned}
&x_w(t+1) = \underbrace{\begin{bmatrix}
      A_f & 0\\[0.3em]
      B C_f $ $ & $ $ A \\[0.3em]
      \end{bmatrix}}_{A_w}x_w(t)
      +\underbrace{\begin{bmatrix}
      B_f\\[0.3em]
      B D_f\\
     \end{bmatrix}}_{B_w}v(t)\\
&y(t)=\underbrace{\begin{bmatrix}
      D C_f $ $ & $ $ C 
      \end{bmatrix}}_{C_w}x_w(t)
      +\underbrace{DD_f}_{D_w}v(t)
\end{aligned}
\end{equation}
  Using \eqref{eq: realizaiton of W}, the state-space model of $W_{ii}$ is given by:  $(A_w, B_w(:,i), C_w(i,:), D_w(i,i))$, where $B_w(:,i)$ is the $i$-th column of $B_w$, $C_w(i,:)$ is the  $i$-th row of $C_w$, and $D_w(i,i)$ is the $(i,i)$-th element of $D_w$. Thus, we construct $O_{\infty}^{W_{ii}}$  based on the state-space realization $(A_w, B_w(:,i), C_w(i,:), D_w(i,i))$ and, for real-time implementation, each SRG uses the state of the {\it entire} system (i.e., $\hat{x}_w=(x_f,\hat{x})$) as feedback. 

 Finally, for the case where the initial conditions are not known, either observer (\eqref{eq:observerClosed} or \eqref{eq:observerClosed2}) can be used to estimate the states; however, during the transient phase of the observer, the states may be incorrect, which may lead to constraint violation. To remedy this issue, one can ``robustify" $O_{\infty}^{W_{ii}}$ as discussed in Section \ref{Sec: Robust DRG}, or alternatively, one could allow the transients to subside before running the system with the reference governor.

\subsection{Identity method}

As previously mentioned, for the identity method, $W(z)$ is either the identity matrix (if $G^{-1}(z)$ is proper) or the identity matrix with one or more time delays (if $G^{-1}(z)$ is not proper). In other words, the input-output behavior of the $i$-th channel is given by $y_i(t) = v_i(t-\beta)$, where $\beta \in \mathbb{Z}_+$ is the delay added to make $G^{-1}(z)$ proper. An interesting observation can be made: the MAS for a pure delay system is independent of the state and  is given by 
\begin{equation} \label{eq:delay MAS}
O_{\infty}^{W_{ii}} = \{(x_{i_0},v_{i_0}): v_{i_0} \in \mathbb{Y}_i\}.
\end{equation}
The above follows directly from the definition of $O_\infty$ in Section \ref{sec:srg} and by noting that the initial states (i.e., previous outputs) of the time-delay system can be chosen as 0, which is automatically admissible. Note also that, MAS for this case is finitely determined, without the need to tighten the steady-state constraint.

The DRG formulation for the case of identity method is the same as \eqref{eq:diagonalRG}, \eqref{eq:kappa for DRG}. However, the implementation is greatly simplified due to the structure of $O_{\infty}^{W_{ii}}$. To see this, note that the structure of \eqref{eq:delay MAS} implies that $\kappa_i$ in \eqref{eq:kappa for DRG} is chosen so that $v_i(t) \in \mathbb{Y}_i$. Since $\mathbb{Y}_i$ is an interval (per Assumption A. \ref{A: intervals}), this implies that $\kappa_i$ is selected so that $v_i(t)$ is simply clipped (i.e., saturated) at the constraint. Thus, the overall DRG can be implemented as a bank of $m$ decoupled saturation functions, which  greatly simplifies real-time implementation.

Similar to the diagonal method, if $G(z)$ has a small condition number or maximal singular value, the inputs to system $G(z)$ would be far away from the references and, hence, tracking performance may suffer. Since this is the same phenomenon as the diagonal case, we will not provide numerical examples.

\begin{figure}
\centering	
\includegraphics[scale=0.6]{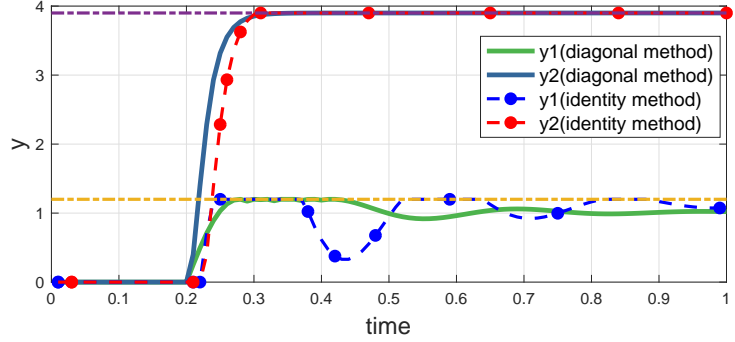}
\caption{Comparison of outputs between diagonal method and identity method.}
\label{fig: oscillation}
\end{figure}

While the identity method is simpler and computationally superior to the diagonal method, it has a drawback. If system $G(z)$ has  under-damped dynamics, then this method would cause large oscillation in the output, even if the plant model is known precisely. To illustrate, we select $q=0.05$ in the example of Section \ref{sec:diag} and change $G_{11}(z)$ in $G(z)$ to the underdamped system: $G_{11}(z)=\frac{0.54z-0.49}{z^{2}-1.85z+0.9}$. A comparison between the outputs of this system after applying DRG with the diagonal and identity methods is shown in Figure~\ref{fig: oscillation}. It  can be seen that the constraints are  satisfied for both outputs. However, unlike the diagonal method, the output using the identity method has large oscillations.

The reason for this behavior can be explained as follows. Because $G(z)$ has slow under-damped dynamics, and since $F^{-1}(z) = G(z)$ for the identity method, applying a step to $r(t)$ causes oscillatory response in $r'(t)$. Viewing DRG as saturations in this case, $v(t)$ is computed as $r'(t)$ clipped at the constraints. Finally, since $W(z)$ is an identity matrix or identity matrix with some time delays, these oscillations will directly show up at the output $y(t)$.

Because of the above shortcoming, it is recommended, before selecting a specific decoupling method, to perform an analysis of the system dynamics similar to the above.

\subsection{Analysis of DRG-tf}\label{sec:analysisDRGtf}
In this section, we present an analysis of DRG-tf, both in  steady-state and transient.

\subsubsection{Steady-State analysis} 

Recall the steady-state constraint in \eqref{eq: O_inf steady state} for a generic system. The steady-state constraint for $O_{\infty}^{W_{ii}}$ can be defined similarly. In order to study the steady-state admissible inputs, we consider the projection of the steady-state constraint onto the $v_i$ coordinate, which results in:
\begin{equation}\label{eq: steady state O_ii}
V_{ss}^{W_{i,i}}:= \{v_i \in \mathbb{R}:W_{{ii}_0}v_i \in \mathbb{Y}_{i,ss}\}
\end{equation}
where $W_{{ii}_0} \in \mathbb{R}$ is the DC gain of subsystem $W_{ii}$ and $\mathbb{Y}_{i,ss}=(1-\epsilon)\mathbb{Y}_i$ (recall that $\mathbb{Y}_i$ is the constraint set for $y_i$). Since $W$ is diagonal, it follows that the steady-state constraint-admissible input set for $W$ is:
\begin{equation}\label{eq: steady state V}
V_{ss}^{W}:=V_{ss}^{W_{1,1}} \times V_{ss}^{W_{2,2}} \times \cdots \times V_{ss}^{W_{m,m}}
\end{equation}
We now compare the above set with the steady-state constraint-admissible input set of system $G$ (projected onto the $u$ coordinate), which arises in VRG applications. This set, noted by  $U_{ss}$, is defined by:
\begin{equation}\label{eq: steady state O_inf}
U_{ss}:=\{u \in \mathbb{R}^m: G_{0}u \in \mathbb{Y}_{ss}\}.
\end{equation}
From the above, the following theorems emerge. Note that Theorems \ref{thm: steady state} and \ref{thm: bounded input} below are also presented in the conference version of this paper (see \cite{DRG2018}). Therefore, we will not present the proofs for brevity, but will present the theorem statements for the sake of completeness.

\begin{theorem}\label{thm: steady state}
For the system of Figure~\ref{fig:Decoupled with RG}, and $U_{ss}$ and $V_{ss}^{W}$ defined in \eqref{eq: steady state V} and \eqref{eq: steady state O_inf}, the following relation holds
\begin{equation}\label{eq: thm 1}
V_{ss}^W=F_0^{-1} \times U_{ss},
\end{equation}
where $F_0$ is the DC gain of $F(z)$ and the operation $F_0^{-1} \times U_{ss}$ is the point-by-point mapping of the set $U_{ss}$ through $F_0^{-1}$.
\end{theorem}

An important implication of this theorem is as follows.  If $r$ is not admissible with respect to system $G$ (i.e., $r\notin U_{ss}$), then $r'$ (see Figure~\ref{fig:Decoupled with RG}) must also not be admissible with respect to the system $W$ (i.e., $r'\notin V_{ss}^{W}$). 

The sets \eqref{eq: steady state V} and \eqref{eq: steady state O_inf} describe the steady-state operations of DRG-tf and VRG, respectively. Note that VRG solves a QP whereas DRG-tf solves an LP. This implies that, for non-admissible references, DRG-tf finds a solution on a vertex of $V_{ss}^{W}$, or from Theorem~ \ref{thm: steady state}, a vertex of $U_{ss}$. On the other hand, VRG finds a solution that may or may not be at a vertex of $U_{ss}$. Therefore, DRG-tf leads to a suboptimal solution with respect to the objective function of VRG.  In the following theorem, we show this more clearly by finding the explicit expression of $v$ computed by DRG-tf at steady-state. For this theorem, recall that the constraint on $y_i(t)$ has the form $\underline{s} \leq y_i \leq \bar{s}$.

\begin{figure}
\begin{center}
\includegraphics[scale=0.2]{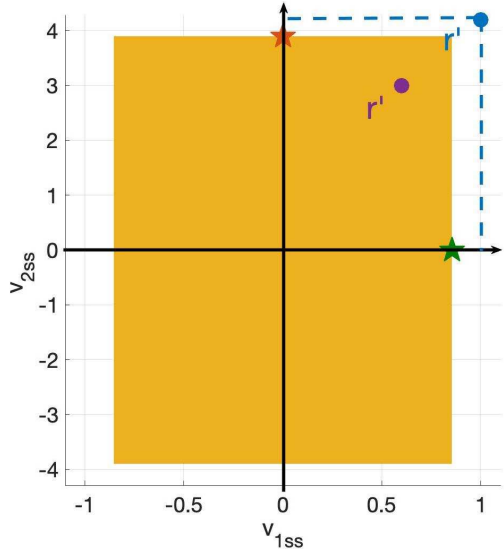}
\caption{Steady-state admissible set for $v_1$ and $v_2$.}
\label{fig: steady state}
\end{center}
\end{figure}

\begin{theorem}\label{thm: thm4}
 For the system of Figure~\ref{fig:Decoupled with RG}, at steady state:
\begin{equation}\label{eq: solution for v_i}
v_i=\emph{\textbf{sat}}(F_{0_i}^{-1}r)
\end{equation}
where $\emph{\textbf{sat}}$ refers to the saturation operator with bounds $W_{ii_{0}}^{-1}\bar{s}$ and $W_{ii_{0}}^{-1}\underline{s}$, and $F_{0_i}^{-1}$ represents the DC gain of the $i$-th row of $F^{-1}$.
\end{theorem} 
\begin{proof}
At steady state, $\underline{s} \leq W_{ii_{0}}v_i \leq \bar{s}$ since $y_i=W_{ii_{0}}v_i$.  If $r_i'$ is constraint admissible, i.e., $\underline{s} \leq W_{ii_{0}}r'_i \leq \bar{s}$, then from \eqref{eq:kappa for DRG} $v_i=r'_i$. Otherwise, $v_i$ would either be equal to $W_{ii_{0}}^{-1}\bar{s}$ or equal to $W_{ii_{0}}^{-1}\underline{s}$. Combining the fact that $r_i'=F_{0_i}^{-1}r$ at steady state, the result follows.
\end{proof}
This theorem shows that if $r_i'$ is in the steady-state admissible set for $v_i$ (i.e., $V_{ss}^{W_{i,i}}$), then $v_i=r_i'$. If this holds for all $i$, then $u=r$ since $F_0F_0^{-1}=I_m$. If $r_i'$ is not in $V_{ss}^{W_{i,i}}$, then $v_i$ can be calculated explicitly as shown in \eqref{eq: solution for v_i}, which means that, from $u=F_0v$, $u$ also can be computed explicitly at steady-state.

We now use an example to show the geometric  interpretation of this theorem. Consider the same example as shown in \eqref{eq: illustrative example1} with $q=0.05$. As before, the constraints are defined by $-1.2 \leq y_1 \leq 1.2$ and $-3.9 \leq y_2 \leq 3.9$. The steady-state admissible set for $v_1$ and $v_2$ is shown by the orange region in Figure~\ref{fig: steady state}. If $r':=(r_1',r_2')=(1,4.2)$ (shown by the blue dot in Figure~\ref{fig: steady state}), which is outside of the admissible set, then $v_1$ is given by the closest point along the $v_{1_{ss}}$ axis to  $r_1'$ (green star in Figure~\ref{fig: steady state}). Similarly, $v_2$ is given by the closest point along $v_{2_{ss}}$ axis to  $r_2'$ (red star in Figure~\ref{fig: steady state}). If $r'$ is in the admissible set (purple dot in Figure~\ref{fig: steady state}), then $v=r'$ and $u=r$ at steady-state.

As previously mentioned, a requirement for DRG is that the plant input, $u$, and the setpoint, $r$, should be equal if no constraint violation is predicted, and that they should be as close as possible if constraint violation is predicted. This is to ensure that the degradation of tracking performance is minimal. We note that each SRG in Figure~\ref{fig:Decoupled with RG} ensures that $v_i$ and $r_i'$ are close; however, $u$ and $r$ may be far. In the following theorem, we show that, at steady-state, the closeness of $u$ and $r$ and, hence, the performance of DRG-tf, depends on the decoupling filter, $F(z)$.

\begin{theorem}\label{thm: bounded input}
Given the system of Figure~\ref{fig:Decoupled with RG}, at steady-state, we have that: 
$$
\|F_0^{-1}\|^{-1} \|v-r'\| \le \|u-r\| \le \|F_0\| \|v-r'\|
$$
where $\|\cdot\|$ refers to any vector norm and its associated induced matrix norm.
\end{theorem}

This theorem shows that $\|u-r\|$ is  bounded above and below by $\|v-r'\|$ scaled by the induced norms of $F_0$ and $F_0^{-1}$, which are known a-priori. More specifically, if $\|F_0\|$ is small, then small $\|v-r'\|$ implies small $\|u-r\|$, which is desirable. Also, if $\|F_0^{-1}\|^{-1}$ is large, then small $\|v-r'\|$ implies large $\|u-r\|$, which is undesirable. In the case of large $\|F_0\|$ or small $\|F_0^{-1}\|^{-1}$, no  conclusion can be made.  

Note that if 2-norm is chosen, then $\|F_0\| = \Bar{\sigma} (F_0)$, where $\Bar{\sigma}(F_0)$ is the largest singular value of $F_0$. Similarly, $\|F_0^{-1}\|^{-1} = \underline{\sigma}(F_0)$, where $\underline{\sigma}(F_0)$ is the smallest singular value of $F_0$. Therefore, 
$$
\underline{\sigma}(F_0) \|v-r'\|_2 \le \|u-r\|_2 \le \Bar{\sigma} (F_0) \|v-r'\|_2.
$$
Since $\|u-r\|_2$ is exactly the objective function in VRG optimization,  the above shows that the performance of DRG-tf and VRG will be close if $F_0$ has small singular values. Note that the quantify $\|v-r'\|$ depends on the value of $r$ and can be computed from Theorem \ref{thm: thm4}.

Finally, note that if the identity decoupling method is implemented, then $F=G^{-1}$. Hence, using Theorem \ref{thm: bounded input}, the following relation follows:
\begin{equation}\label{thm: thm2_2}
\|G_0\|^{-1}\|v-r'\| \le \|u-r\| \le \|G_0^{-1}\| \|v-r'\|,
\end{equation}which allows us to study closeness of $u$ and $r$ using the original system $G(z)$ instead of filter $F(z)$.

\subsubsection{Transient Analysis} 
Here, we extend the steady-state results of the previous section and study the transient performance of DRG-tf. The analysis of this section relies on the ${H_\infty}$ and $L_1$ norm of $F(z)$. Because of the delays introduced in $F(z)$ and/or $F^{-1}(z)$ to make them proper, care must be taken in interpreting the results, as we show below.

\begin{theorem} \label{thm: transient tf}
For the system of Figure~\ref{fig:Decoupled with RG}, the following relationship holds:
\begin{equation}\label{eq: thm 3}
\|u(t+\beta_1)-r(t-\beta_2)\|_{L_2} \leq \|F\|_{H_\infty}\|v-r'\|_{L_2} 
\end{equation}
where $\beta_1$ and $\beta_2$ are the number of delays added to make $F$ and $F^{-1}$ proper, respectively.
\end{theorem}

\begin{proof}
By Parseval's theorem,  \(\|u-r\|_{L_2} = \|U-R\|_{H_2}\) and \(\|v-r'\|_{L_2} = \|V-R'\|_{H_2}\).
where $R'$, $R$, $U$, and $V$ are the $\mathcal{Z}$-transforms of $r'$, $r$, $u$, and $v$, respectively. From Figure~\ref{fig:Decoupled with RG} the following equations hold:
\begin{equation} \label{eq: relationship of u and r}
U(z) = \frac{1}{z^{\beta_1}} F(z)V(z), 
 R'(z)= \frac{1}{z^{\beta_2}}F(z)^{-1}R(z) 
 \end{equation}
Then, 
\[
\begin{aligned}
& \|z^{\beta_1}U(z)-z^{-\beta_2}R(z)\|_{H_2}^2
\\
& = \frac{1}{2\pi}\int_{-\pi}^{\pi}\|F(e^{jw})(V(e^{jw})-R'(e^{jw}))\|_2^2dw\\
& \leq \frac{1}{2\pi}\int_{-\pi}^{\pi}(\|F(e^{jw})\|_2\|V(e^{jw})-R'(e^{jw})\|_2)^2dw  \\
\end{aligned}
\]
where $\|.\|_2$ refers the Euclidean norm. Since $\|F\|_{H_\infty} = \max_{w}\Bar{\sigma}(F(e^{jw}))$, we have that: 
\begin{equation}
\begin{aligned}
 & \frac{1}{2\pi}\int_{-\pi}^{\pi}(\|F(e^{jw})\|_2\|V(e^{jw})-R'(e^{jw})\|_2)^2dw  \\
& \leq \|F\|_{H_\infty}^2\|V-R'\|_{L_2}^2
\end{aligned}
\end{equation}
By Parseval's theorem, the result follows.
\end{proof}
Note that \eqref{eq: thm 3} can be rewritten as: 
\[\|u(t)-r(t-\beta_2-\beta_1)\|_{L_2} \leq \|F\|_{H_\infty}\|v-r'\|_{L_2} \]
This equation shows that the average distance between $u$ and the delayed version of $r$ is bounded by the average distance between $v$ and $r'$ scaled by $\|F\|_{H_\infty}$. Thus, if $\|F\|_{H_\infty}$ is small, the DRG-tf and VRG will perform similarly in transient (although, DRG-tf will exhibit delays).

Since Theorem \ref{thm: transient tf} only discusses time averages, below we provide another theorem to show that the peak of the distance between $u$ and $r$ is related to $\|f\|_{L_1}$, where $f$ is the impulse  response  matrix of  $F(z)$ and $\|f\|_{L_1}$ refers to the $L_1$ norm of $f$. 

\begin{theorem} \label{thm: transient-ss2}
For the system of Figure~\ref{fig:Decoupled with RG}, the following relationship holds with respect to the $L_1$ norm:
\begin{equation}\label{theorem: thm 4}
\|u(t+\beta_1)-r(t-\beta_2)\|_{L_\infty} \leq \|f\|_{L_1}\|v-r'\|_{L_\infty}
\end{equation}

\end{theorem}

\begin{proof}
Based on the inverse $\mathcal{Z}$-transform of \eqref{eq: relationship of u and r}, we have:
\begin{equation}\label{eq6}
\begin{aligned}
&\;\;\;\;\;\; |u_i(t+\beta_1)-r_i(t-\beta_2)|=|f_i(t) * v(t)-f_i(t) * r'(t)| \\
&=|\sum_{\tau=-\infty}^{\infty}\sum_{j=1}^{m}f_{ij}(\tau)(v_j(t-\tau)-r'_j(t-\tau))| \\
&\leq \sum_{\tau=-\infty}^{\infty}\sum_{j=1}^{m}|f_{ij}(\tau)(v_j(t-\tau)-r'_j(t-\tau))| \\
&\leq \|v-r'\|_{L_\infty}\sum_{\tau=-\infty}^{\infty}\sum_{j=1}^{m}|f_{ij}(\tau)|
\end{aligned}
\end{equation}
where $f_{ij}$ refers to the $ij$-th element of $f$,  $*$ denotes the convolution operator, and in the last inequality, we have used the fact that $\|v-r'\|_{L_\infty}$ is the maximal value of $|v_j(t)-r'_j(t)|$ over $j$ and over $t$. 
Taking the maximum of both sides of the above with respect to $i$, we get:
\begin{equation}
\begin{aligned}
\max_{i}|u_i(t+\beta_1)-r_i(t-\beta_2)| \leq \|f\|_{L_1}\|v-r'\|_{L_\infty}\\
\end{aligned}
\end{equation}
and the result follows.
\end{proof}

This theorem implies that if $\|f\|_{L_1}$  is small, then the DRG-tf will perform similarly to VRG in transient. If, however, $\|f\|_{L_1}$ is large, no conclusion can be drawn.

\section{Decoupled Reference Governor Based on State Feedback Decoupling: DRG-ss} \label{section:drg_ss}
In this section, we will introduce DRG-ss and its corresponding steady-state and transient analyses. 
Because DRG-ss uses state feedback decoupling, we assume that  all the states are known or measured. If this is not the case, a standard observer can be designed, which we will not discuss in this paper for the sake of brevity. 
The following assumptions are made for the development of the theory presented in this section:

\begin{assumption}\label{A: system stable2}
 Similar to A. \ref{A: system stable}, 
system $G(z)$ in Figure~\ref{fig: system block ss} is asymptotically stable.
\end{assumption}

\begin{assumption}\label{A: B*}
$B^*$ matrix  in \eqref{eq:B_star} is nonsingular. 
\end{assumption}

In addition, we assume that the sets $\mathbb{Y}_i$ satisfy Assumption A. \ref{A: intervals} in Section \ref{sec: drg-tf}.

Consider the system in Figure~\ref{fig: system block ss}, where we have applied the state feedback decoupling method to get a diagonal system, $W$, which has state space form $(\bar{A},\bar{B},C,0)$ given by
\eqref{eq: Astar and Bstar}. Note that the feedthrough matrix $D$ is taken to be 0 as discussed in Section \ref{Sec: feedback decoupling}, but this assumption can be relaxed. A state-space realization for each decoupled subsystem, $W_{ii}$, is given by: ($\bar{A}$, $\bar{B}(:,i)$, $C(i,:)$, 0), where $\bar{B}(:,i)$ is the $i$-th column of $\bar{B}$, and $C(i,:)$ is the $i$-th row of $C$. Next, for each decoupled subsystem, we compute the MAS, denoted by $O_{\infty}^{W_{ii}}$, as:
\begin{equation}\label{eq:O_inf for DRG_ss}
\begin{aligned}
O_{\infty}^{W_{ii}}:=&\{(x_{w_0},v_{i_0})\in \mathbb{R}^{n+1}: x_{w_0}=x_w(0),\\
& v_i(t)=v_{i_0}, y_i(t)\in \mathbb{Y}_i, \forall t \in \mathbb{Z}_+ \}
\end{aligned}
\end{equation}
where $x_w$ represents the state of $W$. Note that in comparison with DRG-tf, which, depending on the observer design method, may use the states of $W_{ii}$ or $W$ to create $O_\infty^{W_{ii}}$, DRG-ss uses the states of $W$ to create  $O_\infty^{W_{ii}}$.

As for implementation, the SRGs within DRG-ss compute the inputs, $v_i$, to the  decoupled system the same as \eqref{eq:diagonalRG} and $\kappa_i$ is computed by the same linear program as \eqref{eq:kappa for DRG}. Note  that, for the identity decoupling method, the construction of MAS is similar to that of DRG-tf with identity method(see \eqref{eq:delay MAS}); that is, the SRGs can be replaced by a bank of decoupled saturation functions. 

Because of the additional feedback loop (i.e., $-\Phi x$ shown in Figure~\ref{fig: system block ss}), the stability of DRG-ss is not guaranteed (unlike DRG-tf). Below, we provide a sufficient condition for  stability of the DRG-ss scheme.
\begin{figure}
    \centering
        \begin{tikzpicture}
        [node distance=1cm,>=latex']
        \node [input, name=input] {};
        \node [block,minimum size=0.8cm, right=0.8 of sum] (RG) {$Q(z)$};
        \node [block, below of=RG] (feedback) {$SRG$};
        \node [output, right =1cm of RG] (state1) {};
        \node [sum, below=0.8cm of state1] (sum1) {};
        \node [output, below of=state1] (state2) {};
        \node [block,minimum size=0.8cm, right=1cm of sum1] (Ginv) {$\Gamma^{-1}$};
        \node [output, right=1 cm of Ginv] (reference) {};
        \draw [draw,->] (input) -- node[]{} (RG);
        
        \draw [-] (RG) -- node [above,pos=0.79]{}(state1);

        \draw [-] (feedback) -| node[] {} 
        node [near end] {} (input);
        \draw [->] (state1) -- node [left,pos=0.9] {$-$}(sum1);
        \draw [->] (sum1) -- node [] {}(feedback);
        \draw [->] (Ginv) -- node [above,pos=0.79] {$+$}(sum1);
        \draw [->] (reference) -- node [] {}(Ginv);
       
      \draw (1,-0.75)node [color=black,font=\fontsize{10}{10}\selectfont]{$v$};
       
      \draw (2.9,-0.75)node [color=black,font=\fontsize{10}{10}\selectfont]{$r'$};
       
      \draw (6,-0.75)node [color=black,font=\fontsize{10}{10}\selectfont]{$r$};
    \end{tikzpicture}
    \caption{Rearrangement of Figure~\ref{fig: system block ss}.}
    \label{fig: system block ss stable}
\end{figure}
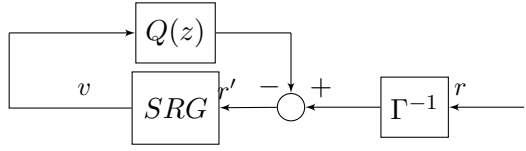 

The block diagram of DRG-ss (Figure~\ref{fig: system block ss}) can be rearranged as shown in Figure~\ref{fig: system block ss stable}, where  $$Q(z)=\Gamma^{-1} \Phi (I- G_x(z) \Phi)^{-1} G_x(z) \Gamma$$
and $G_x(z) = (z I - A)^{-1}B$. 
From Small Gain Theorem (\cite{nonlinearchen}), if there exist four constants $J_1$, $J_2$, $K_1$, and $K_2$, with $J_1 J_2 < 1$, such that:
\begin{equation} \label{eq: small gain theorem}
\centering
\begin{aligned}
     \|v\| \leq K_1 +J_1 \| r' \|, \;\;\;\;\;\;
    \|Q(z)v\| \leq K_2 +J_2 \| v \|
\end{aligned}
\end{equation}
then, the system is bounded input bounded output stable (i.e., BIBO). While $\|\cdot\|$ can be chosen to be any signal norm,  we use the $\infty$-norm in the discussion that follows. Recall that in the SRG optimization  \eqref{eq:kappa for DRG}, $\kappa_i$ satisfies: $0 \leq \kappa_i \leq 1$, which implies that: 
    \[
    \centering
    \begin{aligned}
    \|v(t)\| &= \|v(t-1)+ K(r'(t)-v(t-1)) \|_\infty \\
    & \leq \|(I-K)v(t-1)\|_\infty+\|Kr'(t)\|_\infty \\
    & \leq \|v(t-1)\|_\infty +
    \|r'(t)\|_\infty\\
    \end{aligned}\]
    where $K$ is diagonal matrix with $\kappa_i$ as its main-diagonal elements. Since $v$ is bounded (because $O_\infty^{W_{ii}}$ is compact, see \cite{Gilbert_1991}), we have that $\|v(t-1)\|_\infty \leq M$ for some $M>0$. Thus, 
    \(
    \|v(t)\|_\infty  \leq M + \|r'\|_\infty 
    \)
    (i.e., $J_1 = 1$, $K_1 = M$).  Then, from small gain theorem, the system is BIBO stable if there exist a $K_2$ and $J_2<1$, such that: 
  $
    \|Q(z)v\|_\infty \leq K_2 +  J_2 \| v 
    \|_\infty.
    $
    Recall that the induced system norm  $\|q\|_{L_1}$, where $q$ is the impulse response matrix of $Q(z)$, is defined as:
    \(\|q\|_{L_1} = \sup{\frac{\|Q v\|_\infty}{\|v\|_\infty}}\). 
Then, for $J_2$ to exist, the following inequality needs to be satisfied:
\[\|q\|_{L_1} < 1\]

In summary, the DRG-ss scheme is BIBO stable if $\|q\|_{L_1} < 1$. 
It is important to note that $Q(z)$ depends on $\Gamma$ and $\Phi$. Thus, stability must be checked after $\Phi$ and $\Gamma$ have been designed, which means that iterations might be needed if the stability condition above is not satisfied. Finally,  asymptotic stability can also be proved by applying the results from absolute stability (\cite{harris1983stability}) to the system of Figure~\ref{fig: system block ss stable} and using the fact that $0 \leq \kappa_i\leq 1$.

\begin{figure}
\centerline{\includegraphics[scale=0.45]{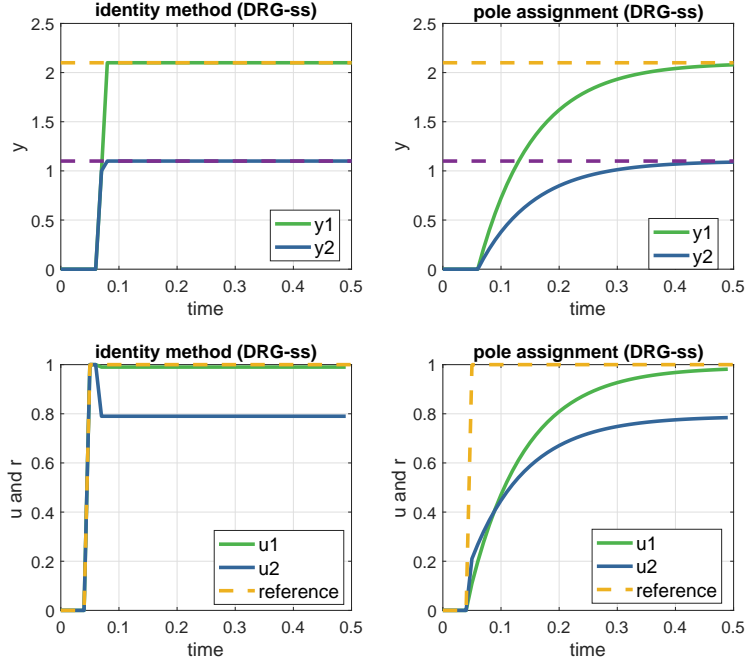}}

    \caption{Simulation results of DRG-ss. The purple and yellow dashed lines on the top two plots represent the constraints on the outputs.}
\label{fig.drgss}
\end{figure}

Next, we provide an example for DRG-ss, where the two decoupling methods in Section \ref{Sec: feedback decoupling} are  applied to decouple the system. Consider the system $G$ given by:
\begin{equation}\label{eq: example drg-ss}
A=\begin{bmatrix}
      0.1& 1& 0\\
      0 & 0.1& 0\\
      0 &0 &0.1\\
     \end{bmatrix},B=\begin{bmatrix}
      0&\quad 1\\
      1&\quad 0\\
      1&\quad 0\\
     \end{bmatrix},C=\begin{bmatrix}
      1 &\quad1 &\quad-1\\
      0 &\quad1 &\quad0\\
     \end{bmatrix}
     \end{equation}
We use \eqref{eq:pair1} and \eqref{eq:pair2} to find $\Phi$ and $\Gamma$, and proceed to compute  $O_{\infty}^{W_{11}}$ and $O_{\infty}^{W_{22}}$ based on \eqref{eq:delay MAS} and \eqref{eq:O_inf for DRG_ss} (for the identity and pole assignment methods, respectively). Note that for pole assignment method, we choose $M_k=\mathrm{diag}(0.9,0.9)$ to locate two of the poles of $W$ at $0.9$. The constraint set is defined as $\mathbb{Y}:=\{ (y_1, y_2): y_1 \leq 2.1, y_2 \leq 1.1 \}$. We simulate the response of this system to a step of size $1$ in both $r_1$ and $r_2$. The simulation results  are depicted in Figure~\ref{fig.drgss}.

Figure~\ref{fig.drgss} (top) shows that the outputs are within the constraints for both identity and pole assignment methods. Note, from the bottom plots of Figure~\ref{fig.drgss}, that there is a gap between $u$ and $r$. Later, we will investigate this gap.

As a final remark, similar to the identity method for DRG-tf, while the identity method for DRG-ss is simpler and computationally superior to the pole assignment method, it has a drawback: it may lead to large oscillations for underdamped systems.


\subsection{Analysis of DRG-ss}\label{Sec:analysis of FRG-ss}

In this section, we present steady-state and transient analyses of DRG-ss. 

\subsubsection{Steady-state Analysis} 
We begin by noting that the definitions of the steady-state halfspace for DRG-ss, i.e., $V_{ss}^W$, and VRG, i.e.,  $U_{ss}$, are the same as  \eqref{eq: steady state V} and \eqref{eq: steady state O_inf}. Below, we  present a theorem to relate $U_{ss}$ and $V_{ss}^W$, which parallels Theorem \ref{thm: steady state} for DRG-tf.

\begin{theorem}\label{thm: steady state ss}
For the system of Figure~\ref{fig: system block ss}, and $U_{ss}$ and $V_{ss}^{W}$ defined in \eqref{eq: steady state O_inf} and \eqref{eq: steady state V}, the following relation holds
\begin{equation}\label{eq: thm 5}
V_{ss}^{W}=C(I-\bar{A})^{-1}\bar{B} (C(I-A)^{-1}B)^{-1} \times U_{ss}
\end{equation}
where $\bar{A}=A + B \Phi$ and $\bar{B} = B \Gamma$. 
\end{theorem}

\begin{proof}
Given the state-space realization $(A,B,C,0)$ for $G(z)$,  the DC-gain of $G$ from $u$ to $y$ is given by  $G_0=C(I-A)^{-1}B$. Similarly, the DC-gain of $W$ from $v$ to $y$ is given by $W_0 = C(I-(A+B \Phi))^{-1}B\Gamma$. Therefore, the relationship between $W_0$ and $G_0$ is as follows: 
 \[W_{0}=C(I-(A+B \Phi))^{-1}B \Gamma (C(I-A)^{-1}B)^{-1} \times G_0\]
The proof follows from the definitions of $U_{ss}$ and $V_{ss}^{W}$.
\end{proof}

This theorem shows that if $r$ is not admissible with respect to system $G$ (i.e., $r\notin U_{ss}$), then, after feeding through $\Gamma^{-1}$, $r'$ must also not be admissible with respect to the system $W$ (i.e., $r'\notin V_{ss}^{W}$). 

Before, we mentioned one requirement for DRG, which was $u$ and $r$ should be as close as possible. From Figure~\ref{fig: system block ss}, we see that $v$ and $r'$ are as close as possible, but $u$ and $r$ may not be close. Below, we provide a theorem to quantify the closeness of $u$ and $r$ in steady state.

\bigskip
\begin{theorem}\label{thm: steady state ss2}
For the system of Figure~\ref{fig: system block ss}, the following relation holds at steady state:
$$
\|\Gamma^{-1}\|^{-1} \|v-r'\| \le \|u-r\| \le \|\Gamma\| \|v-r'\|
$$
where $\|.\|$ refers to any vector norm and its associated induced matrix norm.
\end{theorem}
\begin{proof}
At steady state, we have that $u=\Gamma v + \Phi x$ and $r=\Gamma r' + \Phi x$. Therefore: $\|u-r\|= \|\Gamma v-\Gamma  r'\|=\|\Gamma (v-r')\|\leq \|\Gamma\| \|v-r'\|.$ This proves the right hand inequality. To show the left hand inequality, write $\|v-r'\|= \|\Gamma^{-1} u-\Gamma^{-1} r\|=\|\Gamma^{-1} (u-r)\|\leq \|\Gamma^{-1}\| \|u-r\|.$ This can be re-written as $\|\Gamma^{-1}\|^{-1} \|v-r'\| \le \|u-r\|$, which concludes the proof.
\end{proof}
 This theorem shows that  $\|u-r\|$ is  bounded above and below by $\|v-r'\|$ scaled by $\|\Gamma\|$ and $\|\Gamma^{-1}\|^{-1}$, which are known {\it a-priori}. More specifically, if $\|\Gamma\|$ is small, then small $\|v-r'\|$ implies small $\|u-r\|$, which is desirable. Also, if $\|\Gamma^{-1}\|^{-1}$ is large, then small $\|v-r'\|$ implies large $\|u-r\|$, which is undesirable. In the case of large $\|\Gamma\|$ or small $\|\Gamma^{-1}\|^{-1}$, no definite conclusion can be made. Note that the steady-state analysis of $v$ is similar to that in DRG-tf (see Theorem \ref{thm: thm4}), except that instead of having $r'=F_0^{-1}r$ in DRG-tf, we have $r'=\Gamma^{-1}(r-\Phi x)$ in DRG-ss. For the sake of brevity, we will not provide the detailed analysis in this section.

\begin{remark}
Similar to DRG-tf, DRG-ss may compute $u_i$ to be larger or smaller than $r_i$ depending on the matrix $\Gamma^{-1}$. Note that $u_i > r_i$ may or may not be desirable, as we discussed in Section \ref{sec: drg-tf}. 
\end{remark}

\subsubsection{Transient Analysis} 

Recall from Figure~\ref{fig: system block ss} that the following relationship holds:
\begin{equation}\label{eq: relationship ss}
r' = \Gamma^{-1}(r-\Phi x), \;\;\;
u = \Gamma v+\Phi x 
\end{equation}
From these equations, the following theorem emerges, which discusses the transient performance of DRG-ss:

\begin{theorem}\label{thm: transient ss}
For the system in Figure~\ref{fig: system block ss}, the following inequalities hold:
\begin{equation}\label{eq: drg-ss transient1}
\|u-r\|_{L_2} \leq  \sqrt{ \sum_{i,j} \Gamma_{ij}^2 }   \times \|v-r'\|_{L_2} \end{equation}
\begin{equation}\label{eq: drg-ss transient2}
\|u-r\|_{L_\infty} \leq m \times \max_{i,j}  |\Gamma_{ij}| \times \|v-r'\|_{L_\infty} 
\end{equation}
where $\Gamma_{ij}$ is the $ij$-th element of $\Gamma$.
\end{theorem}
\begin{proof}
From \eqref{eq: relationship ss}, the following equation holds: $u-r = \Gamma (v-r')$. Then,
\[
\|u - r\|_{L_2}^2 = \|\Gamma (v-r')\|^2_{L_2}
= \sum_{t=0}^{\infty}\sum_{i=1}^{m} (\Gamma_i (v(t)-r'(t)))^2 
\]
where $\Gamma_i$ refers to the $i$-th row of $\Gamma$. By Cauchy-Schwarz inequality, we have:
\[
\begin{aligned}
&\;\;\;\;\;\;\;\; \sum_{t=0}^{\infty}\sum_{i=1}^{m} (\Gamma_{i}(v(t)-r'(t)))^2
\leq \sum_{t=0}^{\infty}\sum_{i=1}^{m} \| \Gamma_{i}\| \|v(t)-r'(t)\| \\
& \leq \sum_{i=1}^{m} \| \Gamma_{i}\| \sum_{t=0}^{\infty} \|v(t)-r'(t)\| 
 =  \sum_{i,j} \Gamma_{ij}^2 \| (v-r') \|^2_{L_2}
\end{aligned}
\]
Taking the square root of both sides proves \eqref{eq: drg-ss transient1}. Next, we will show the proof of \eqref{eq: drg-ss transient2}. We have that:
\[
\begin{aligned}
&\;\;\;\;\;\;\;\; \|u - r\|_{L_\infty} = \|\Gamma (v-r')\|_{L_\infty} = \sup_{t \geq 0}(\max_{i}|\Gamma_i (v-r')|)\\
 &\leq \sup_{t \geq 0} (\max_{i,j} |m \Gamma_{ij}|) (\max_{i}| (v-r')|)  = \max_{i,j} |m \Gamma_{ij}| \| (v-r') \|_{L_\infty}
\end{aligned}
\]
Then, \eqref{eq: drg-ss transient2} follows.
\end{proof}
The theorem presents the relationship between $v-r'$ and $u-r$ and shows that if the elements of $\Gamma$ are small, then the distance between $u$ and $r$ would also be small. This implies that tracking will not be significantly deteriorated as compared with VRG.

\section{Computational Considerations}\label{sec: computation time}
In this section, we discuss the computational aspects of DRG and compare the run-time of DRG with VRG. Since the filters $F(z)$ and $F^{-1}(z)$ in DRG-tf and the matrix multiplications in DRG-ss can be implement easily, we will not focus on them. The focus of this section will instead be on the implementation of the SRGs that are used in the DRG formulation. Note that the SRGs in DRG-tf and DRG-ss are the same, so we will only consider DRG-tf in this section. 

Recall that the implementation of the DRG on an $m$-input $m$-output system involves solving $m$ linear programs (LP), described by \eqref{eq:kappa for DRG}.  These LPs can be solved implicitly via LP solvers, or explicitly as explained below.  VRG, on the other hand, requires the solution to a Quadratic Program (QP), which can be solved  implicitly via online optimization or explicitly via multi-parametric programming. In this work, we use the MPT Toolbox  in Matlab to implement implicit QP and implicit LP (MPT was the fastest among other solvers such as Gurobi). Also, we use the algorithm that is introduced in \cite{tondel2003algorithm} to implement explicit QP.

For the explicit DRG mentioned above, we implement Algorithm \ref{Algorithm:algorithm1}, which provides an algorithm to compute $\kappa_i$ in \eqref{eq:kappa for DRG} for each SRG. For this algorithm, we have assumed that $O_{\infty}^{W_{ii}}$ is given by polytopes of the form \eqref{eq:OinfForm}, and that $j^{*}$ denotes the number of rows of $H_x, H_v, h$. Note that we have used the notation $H_v$ instead of $H_u$ because the output of the SRGs in DRG are $v_i$ and not $u_i$. In this algorithm, with some abuse of notation, we use $x$ to refer to the state that is fed back to the $i$-th SRG (i.e., either the state of the $i$-th subsystem or the state of the entire system as explain in Section \ref{sec:DRG-tf observer}).  

\begin{algorithm}
 \caption{Custom Explicit DRG Algorithm}
 \begin{algorithmic}[1]
 \STATE let $a=H_{v}(r_i'(t)-v_i(t-1))$
 \STATE let $b=h-H_{x}x(t)-H_{v}v_i(t-1)$
 \STATE set $\kappa=1$
  \FOR {$i = 1$ to $j^*$}
  \IF {$a(i) > 0$}
  \STATE $\kappa=\min(\kappa,b(i)/a(i))$
  \ENDIF
  \ENDFOR
  \STATE $\kappa_i=\max(\kappa,0)$
 \end{algorithmic} 
 \label{Algorithm:algorithm1}
 \end{algorithm}

To compare the performance of DRG with VRG, we use an example of a distillation process, which is a two-input and two-output coupled system presented in \cite{Skogestad_2007}. The DRG formulation for this system requires the solution to two LPs, whereas the VRG formulation requires the solution to a single QP. All simulations were performed in Matlab R2017b. The simulation device  is a Macbook with 1.1 GHz Intel Core m3 processor and 8 GB memory.

We simulate the distillation process using 4 different governor/solver combinations: explicit DRG (i.e., Algorithm 1), implicit DRG (i.e., implicit LP), explicit VRG (i.e., explicit QP), and implicit VRG (i.e., implicit QP). The simulation length is 10000 time steps in all cases with a sample time of 0.01s. Upon simulating the system, we compute the average and maximum computation times of the solvers. In order to eliminate the effects of background processes running on the computer, each of the above experiments are run 5 times and the averages are computed. 
The results are shown in Table \ref{table: computation time for VRG} and Table \ref{table: computation time for DRG}.  As can be seen, the average time indicates that the Explicit RG is two orders of magnitude faster than explicit VRG and explicit VRG runs three orders of magnitude faster than the rest of the governors, which means that DRG computation terminates  faster than VRG.

\begin{table}
\renewcommand{\arraystretch}{1.3}
\caption{Computation time for VRG in the practical example.}
\centering
\begin{tabular}{||c||c||c||}
\hline
& \bfseries Explicit QP & \bfseries Implicit QP\\
\hline\hline
average & $8.71\times 10^{-5}$s & $0.45 \times 10^{-2}$s\\
\hline
\hline
maximum & $7.66 \times 10^{-4}$s & $2.3 \times 10^{-2}$s\\
\hline
\end{tabular}
\label{table: computation time for VRG}
\end{table}

\begin{table}
\renewcommand{\arraystretch}{1.3}
\caption{Computation time for DRG in the practical example}
\label{table: computation time for DRG}
\centering
\begin{tabular}{||c||c||c||}
\hline
& \bfseries Implicit LP & \bfseries Algorithm 1\\
\hline\hline
average & $0.49 \times 10^{-2}$s & $5.2 \times 10^{-7}$s\\
\hline
\hline
maximum & $2.2 \times 10^{-2}$s & $1.42 \times 10^{-5}$s\\
\hline
\end{tabular}
\end{table}

\section{Robust DRG}\label{Sec: Robust DRG}
In section \ref{sec:MAS with disturbance}, we provided a brief explanation of how SRG can be modified to handle systems  affected by unknown disturbances and sensor noise. Essentially, MAS is ``robustified" (i.e., shrunk) to account for the worst-case realization of the disturbances. In this section, we extend these ideas to DRG-tf and DRG-ss, where we show that an initial pre-processing is required to have the system in the form \eqref{eq:system with disturbance}. Secondly, we consider the case where the system model is uncertain, where we present an innovative solution for handling these systems.

\subsection{DRG for Systems with Unknown Disturbances}\label{sec:DRG for disturbance}

\subsubsection{DRG-tf for systems with unknown disturbances}
Suppose system \eqref{eq:Gz} is now affected by an unknown disturbance $d(t)\in \mathbb{R}^d$:
\begin{equation}\label{eq:Gz with disturbance}
Y(z)=G(z)U(z)+G_w(z)D(z)
 \end{equation}
where $D(z)$ is the $\mathcal{Z}$-transform of $d(t)$. Consistent with the literature of SRG, it is  assumed that $d \in \mathbb{D}$, where $\mathbb{D}$ is a compact polytopic set.

In this section, we consider DRG-tf with the diagonal decoupling method explained in Section \ref{sec: drg-tf} (the identity decoupling method can be applied similarly). Under  Assumption  A.\ref{A: NMP}, we compute the filter $F(z)$ defined in \eqref{eq:diagonal method}. This leads to each $y_i$ described by: $Y_i(z)=G_{ii}(z)V_i(z)+\sum_{j=1}^{d} G_{w_{ij}}(z) D_j(z)$, which is decoupled from $v$ to $y$, but not from $d$ to $y$. 
To address this, we convert the dynamics of each $y_i$ to state-space form:
\begin{equation}\label{eq:model state space for DRG-tf}
\begin{aligned}
&x_i(t+1)=A_ix_i(t)+B_i v_i(t)+B_{w_i}d(t)\\
&y_i(t)=C_ix_i(t)+D_{w_i}d(t) \in \mathbb{Y}_i
\end{aligned}
\end{equation}
For each subsystem \eqref{eq:model state space for DRG-tf} we now proceed to compute the corresponding robust MAS using the procedure described in Section \ref{sec:MAS with disturbance}. The implementation of DRG-tf is otherwise unchanged.

\subsubsection{DRG-ss for systems with unknown disturbances}\label{sec: DRG-ss with dist}

\begin{figure}
\centerline{\includegraphics[scale=0.3]{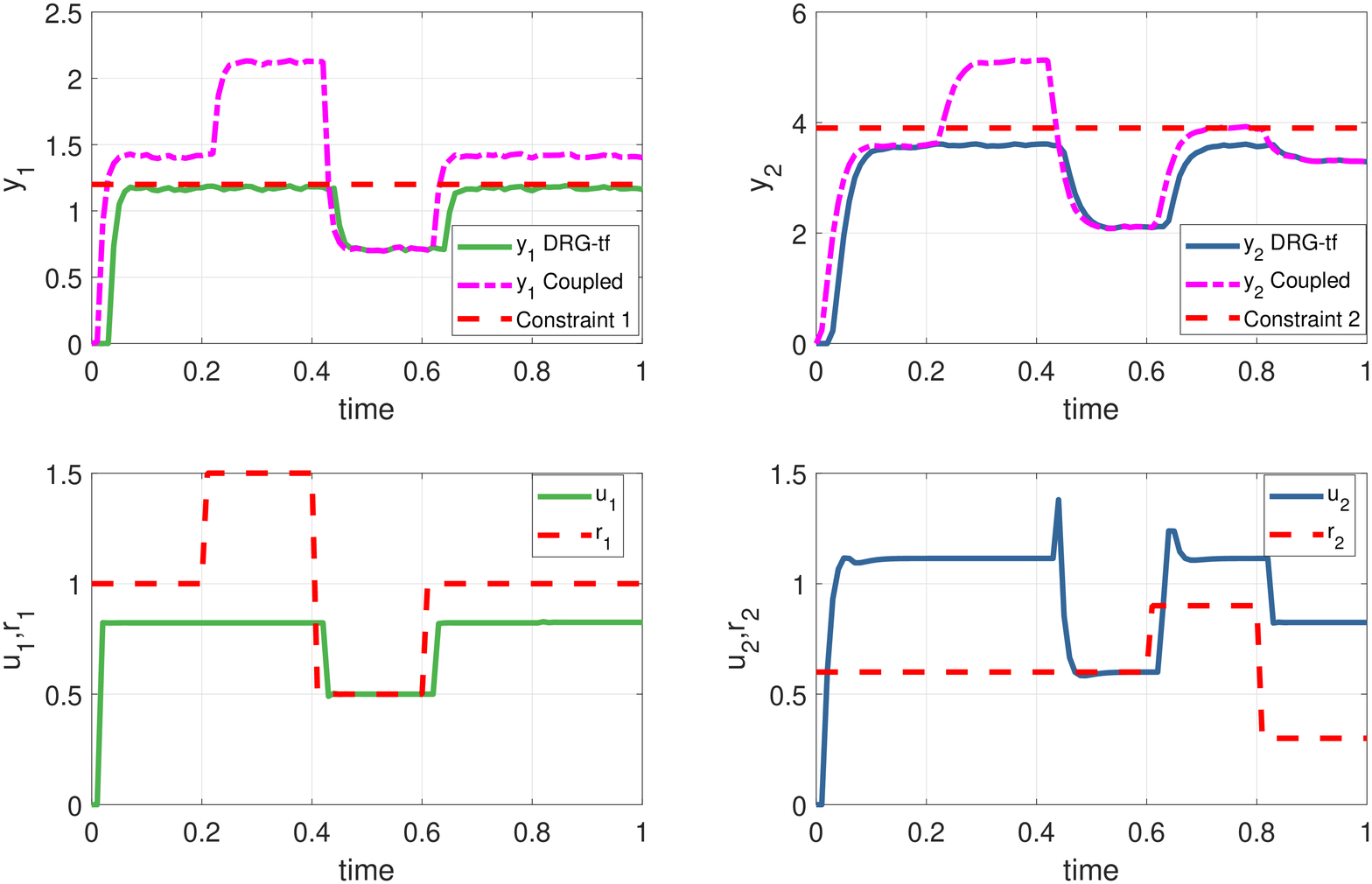}}
    \caption{DRG-tf with disturbance.}
\label{fig:DRG-tf example}
\end{figure}

In order to decouple system \eqref{eq:system with disturbance} from the inputs $u$  to the outputs $y$, we apply the pole assignment decoupling method  explained in Section \ref{section:drg_ss}; similar results can be obtained for the identity decoupling method. The decoupled system to consider is:
\begin{equation}\label{eq:system with disturbance and decoupling}
\begin{aligned}
&x(t+1)=(A+B\Phi)x(t)+B\Gamma v(t)+B_{w}d(t)\\
&y(t)=Cx(t)+D_{w}d(t) \in \mathbb{Y}
\end{aligned}
\end{equation}
where $\Phi$ and $\Gamma$ are computed based on \eqref{eq:pair2}, and $v$ is the input obtained from the SRGs (see Figure~\ref{fig: system block ss}). The $i$-th decoupled subsystem can then be written as: 
\begin{equation}\label{eq:individual systems with disturbance and decoupling}
\begin{aligned}
&x(t+1)=\bar{A}x(t)+B_iv_i(t)+B_{w}d(t)\\
&y_i(t)=C_ix(t)+D_{w_i}d(t) \in \mathbb{Y}_i
\end{aligned}
\end{equation}
where $\bar{A}=A+B\Phi$, $B_i$ is the $i^{th}$ column of $B\Gamma$, $C_i$ is the $i^{th}$ row of $C$, and $D_{w_i}$ is the $i^{th}$ row of $D_w$. Based on \eqref{eq:individual systems with disturbance and decoupling} we  create the corresponding robust MAS for the $i$-th subsystem.  The DRG-ss implementation is otherwise unchanged.

Next, we will illustrate the above ideas with two examples, one for DRG-tf and another for DRG-ss. Both examples are necessary in order to highlight the subtleties of the two approaches.

\begin{figure}
\centerline{\includegraphics[scale=0.3]{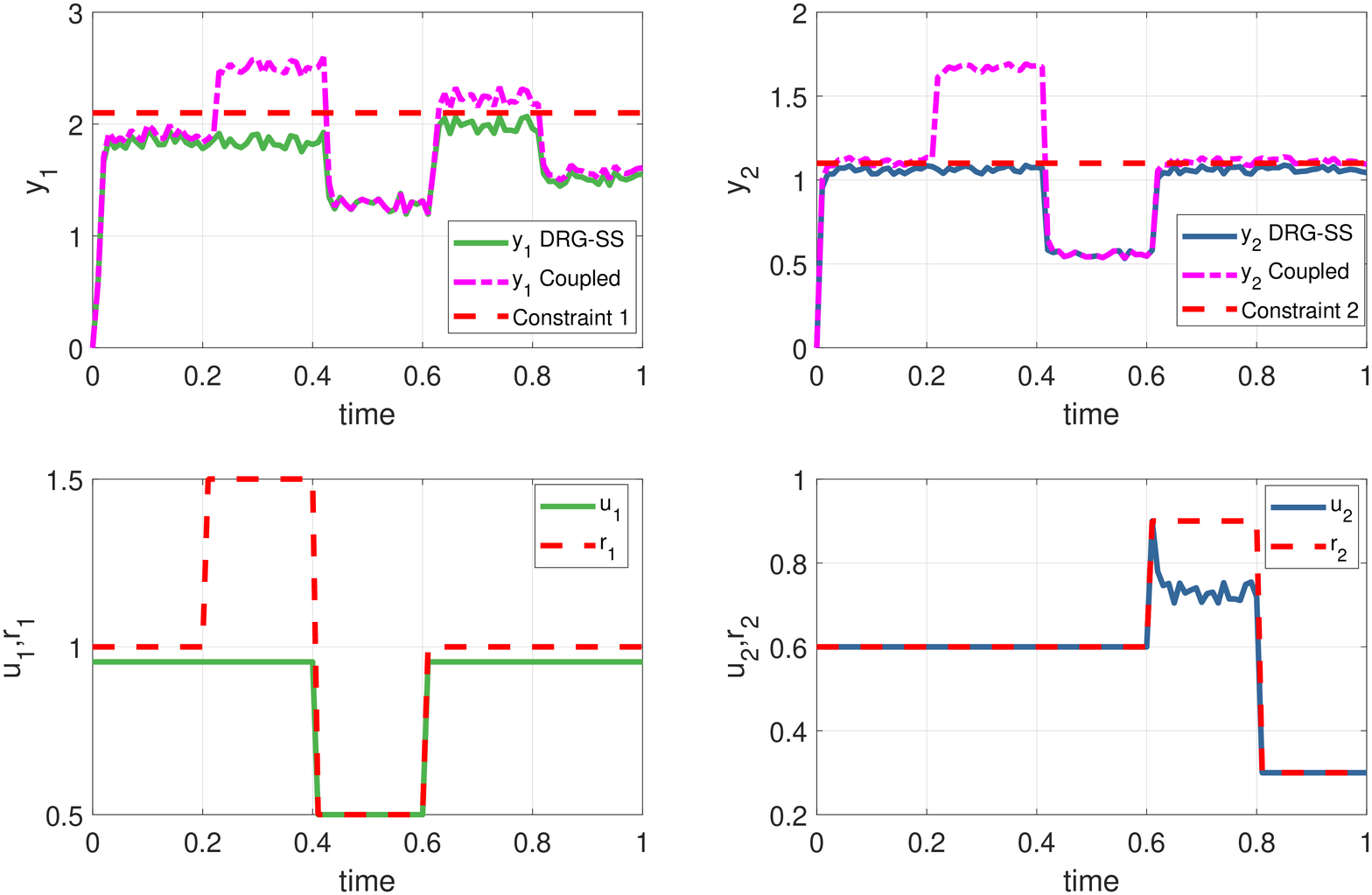}}
    \caption{DRG-ss with disturbance.}
\label{fig:DRG-ss example}
\end{figure}

For DRG-tf, we consider the system \eqref{eq:Gz with disturbance} with $q=0.05$ and $G_w(z)=\begin{bmatrix}
      \frac{0.2}{(z-0.5)^2(3z+1)}\\[0.3em]
      \frac{0.3}{(2z+1)(z-0.7)^2}\\[0.3em]
\end{bmatrix}$, and the constraints: 
$-1.2 \leq y_1\leq 1.2, -3.9 \leq y_2\leq 3.9$. 
We implement DRG-tf for this system assuming the  disturbance  satisfies $d\in \mathbb{D}:=[-0.1, 0.1]$. For DRG-ss, we consider again system \eqref{eq: example drg-ss} used in Section \ref{section:drg_ss} with the   output constraints:
$y_1\leq 2.1,    y_2\leq 1.1$. Assume $D_w$  is zero,  $B_w=[1.3,0.3,2.51]^\top$, and that the disturbance also satisfies $d(t)\in \mathbb{D}:=[-0.1, 0.1]$. We decouple the system using the pole assigment method, placing the closed-loop poles at 0.1. For the purpose of simulations, the disturbance in both cases is generated randomly and uniformly from the interval $[-0.1, 0.1]$.

The results of DRG with disturbance are shown in Figures \ref{fig:DRG-tf example} and \ref{fig:DRG-ss example}. In the top subplots of these figures, ``$y_1$ coupled" and ``$y_2$ Coupled" refer to the response of the system without DRG (i.e., $r$ applied to $G$ directly), which shows that, without a DRG, the  constraints are violated. These results confirm that DRG is able to satisfy the constraints in the presence of disturbances. As can be seen from the plots, the disturbance affects both outputs (the outputs appear noisy). Interestingly, the disturbance does not affect $u$ for DRG-tf (see Figure~\ref{fig:DRG-tf example}), but it affects $u$ for DRG-ss (see $u$ in Figure~\ref{fig:DRG-ss example}). The reason for this behavior can be explained as follows: it can be seen from Figure~\ref{fig: system block ss} that the outer feedback in  DRG-ss may transmit the effects of disturbances and sensor noise  to $r'$. As a result of this, the effect of the disturbance on the output may be higher in DRG-ss than in DRG-tf. This may be a decisive argument to select between  DRG-ss and DRG-tf, since the latter does not show this type of behavior. 

\begin{remark}
For a system in which the states are not measured, a standard observer may not provide accurate estimation of the state if unknown disturbances affect the system. In such a case, we refer to the work developed in \cite{kalabic_2015}, where an observer which considers the error introduced by unknown disturbances is implemented. 
\end{remark}

\subsection{DRG with parametric uncertainty}
In this section, we briefly sketch the approach that can be used  for cases when  system $G(z)$ in Figures \ref{fig:Decoupled with RG} and  \ref{fig: system block ss} has parametric uncertainty, that is, matrices $A$ and $B$ are uncertain or vary in time. For simplicity, we assume matrix $C$ is known and $D=0$. The approach we take is similar to  \cite{kerrigan2001robust}. Note that we consider parametric uncertainties in the state-space matrices, because the RG approach is a time-domain approach. Therefore, frequency domain uncertainties  are not investigated. We assume that the uncertain/time-varying closed-loop system (i.e., $G(z)$) is asymptotically stable. Therefore, stability is still not a concern in DRG-tf, but additional analysis must be carried out to ensure stability of DRG-ss. This is similar to our prior discussion in Section \ref{section:drg_ss} so we will not dwell on the issue of stability.

For this discussion, reconsider system $G(z)$, but now with parametric uncertainty on the $A$ and $B$ matrices, which leads to the square linear system given by:
\begin{equation}\label{eq:system with parametric uncertainty}
\begin{aligned}
&x(t+1)=A(t)x(t)+B(t)u(t)\\
&y(t)=Cx(t) \in \mathbb{Y}
\end{aligned}
\end{equation}
In \cite{kerrigan2001robust}, in order to compute the robust MAS for this type of systems, it is assumed that the pair $(A(t),B(t))$ belongs to a given uncertainty polytope defined by the convex hull of the matrices $(A^{(j)},B^{(j)})$, that is 
$$
(A(t),B(t)) \in \mathrm{conv} \{(A^{(1)},B^{(1)}), \ldots, (A^{(N)}, B^{(N)}) \},
$$  
where $N$ is the number of  vertices in the uncertainty polytope (\cite{Pluymers_2005}). Applying this idea directly to DRG, however, may not guarantee constraint satisfaction because the parametric uncertainties  will prevent us from perfectly decoupling the system. To explain, suppose we select a nominal pair of $A$ and $B$ matrices from the convex hall, and decouple this nominal system by computing the matrices $\Phi$ and $\Gamma$ using \eqref{eq:pair1} or \eqref{eq:pair2}. Since the matrices of the actual system will be different from the nominal ones, this decoupling process results in:
\begin{equation}\label{eq:decoupled with parametric}
\begin{aligned}
x(t+1)=\bar{A}(t)x(t)+\bar{B}(t)v(t),\;\;\;\; y(t)=Cx(t)
\end{aligned}
\end{equation}
where the pair $(\bar{A}(t),\bar{B}(t))$  satisfies:
 \begin{equation}\label{eq:conv uncertainty pair}
(\bar{A}(t),\bar{B}(t)) \in \mathrm{conv} \{(\bar{A}^{(1)},\bar{B}^{(1)}), \ldots, (\bar{A}^{(N)}, \bar{B}^{(N)}) \},
 \end{equation} 
where $\bar{A}^{(j)}=A^{(j)}+B^{(j)}\Phi$, $\bar{B}^{(j)} =B^{(j)} \Gamma$. Clearly, these dynamics are not decoupled for all matrices in the uncertainty polytope. This implies that DRG implemented on \eqref{eq:decoupled with parametric} may not achieve perfect decoupling and thus may not enforce the constraints.

To address the above problem, we introduce a novel margin in each $O_{\infty}^{W_{ii}}$ to robustify each channel against these coupling dynamics. To explain, consider the dynamics of the $i$-th output of \eqref{eq:decoupled with parametric}:
\begin{equation}\label{eq:decouple unt using dist}
\begin{aligned}
&x(t+1)=\bar{A}(t)x(t)+\bar{B_i}(t)v_i(t)+B_w(t)\bar{v}(t)\\
&y_i(t)=C_ix(t)
\end{aligned}
\end{equation}
where $C_i$ is the $i$-th row of $C$,  $\bar{B}_i(t)$ corresponds to the $i^{th}$ column of $\bar{B}(t)$,  $B_w(t)$ gathers all columns of $\bar{B}(t)$ except the $i^{th}$ one, and $\bar{v}(t)$ represents the vector containing all inputs except the $i$-th one, i.e., vector of all $v_k$'s, $k \neq i$. Our solution below treats $\bar{v}$ as an unknown bounded disturbance. To accomplish this, we quantify a lower and an upper bound on $\bar{v}$ and robustify $O_\infty^{W_{ii}}$ using results similar to Section \ref{sec:DRG for disturbance}. Specifically, to find the bounds, we leverage the fact that each element of $\bar{v}(t)$, $\bar{v}_k$,  is the output of an SRG, whose goal is to enforce the constraints on the $k$-th output (i.e., $y_k(t)\in \mathbb{Y}_k$). Thus,  we can define upper and lower bounds on each element of $\bar{v}$ using the steady-state constraints \eqref{eq: steady state O_ii}:
\begin{equation}
\begin{aligned}
 &   \bar{v}_k^{\max}=\max\{\bar{v}_{k}:W_{kk_0}^{(j)}\bar{v}_{k}\in (1-\epsilon)\mathbb{Y}_k, j = 1,\ldots,N\}\\
 &   \bar{v}_k^{\min}=\min\{\bar{v}_{k}:W_{kk_0}^{(j)}\bar{v}_{k}\in (1-\epsilon)\mathbb{Y}_k, j = 1,\ldots,N\}
\end{aligned}
\end{equation}
 where  $W_{kk_0}^{(j)}$ represents the DC gain of the system from the $k$-th input to the $k$-th output given the pair $(\bar{A}^{(j)},\bar{B}^{(j)})$. Since we have that each ${\bar{v}_k(t)\in [\bar{v}_k^{\min}, \bar{v}_k^{\max}]}$, we can now treat $\bar{v}(t)$ in \eqref{eq:decouple unt using dist} as an unknown bounded disturbance to create a robust MAS set for the $i$-th channel, which can be accomplished using the ideas from Section \ref{sec:DRG for disturbance} (for unknown disturbances) and references \cite{kerrigan2001robust,Pluymers_2005} (for polytopic uncertainties). Implementation of DRG using these MAS's will ensure that the system is robust to the plant/model mismatch and, thus, the constraints will be satisfied. It is important to mention that this approach may  lead to conservative results depending on  how much the MAS is shrunk. However, if the system is ``almost" decoupled (i.e., the nominal system is close to the actual one), then the shrinkage will be negligible. 
For the sake of brevity, numerical examples and further analysis on this topic will appear in our future work.

\section{Extension of DRG to non-square MIMO systems}\label{Sec: nonsquare system}

In this section, we will briefly introduce the extension of DRG to non-square MIMO systems, i.e., systems where the number of inputs is either larger or smaller than the number of outputs. We will treat these cases separately in the following subsections. Generally speaking, we achieve this by either introducing fictitious outputs to transform the system into a square one (see Figure~\ref{fig:DRG-tf nonsquarelarger}), or only decoupling a square subsystem of it (see Figure \ref{fig: DRG-tf non-square large y}). For the sake of clarity, we will only focus on the extension of DRG-tf with the diagonal method; the same process can be applied to DRG-tf with identity method and DRG-ss.

\subsection{Systems with larger number of inputs}

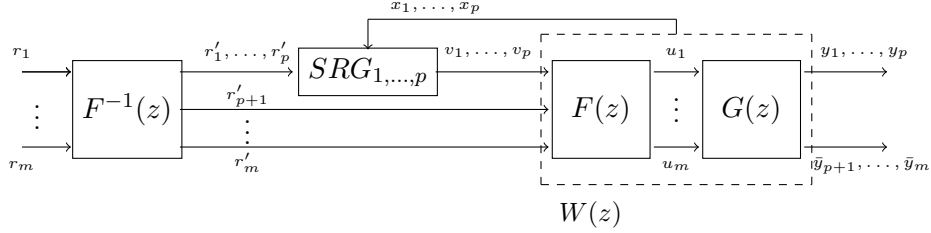
\begin{figure}
\centering
\begin{tikzpicture}
    [L1Node/.style={rectangle,draw=black,minimum size=5mm},
    L2Node/.style={rectangle,draw=black, minimum size=13mm},
    L3Node/.style={rectangle,draw=black, minimum size=30mm}]
      \node[L2Node] (n3) at (1.4, 0){$F^{-1}(z)$};
      \node[L1Node] (n4) at (4.6, 0.5){$SRG_{1,\ldots,p}$};
      \draw (3,-0.2)node [color=black,font=\fontsize{10}{10}\selectfont]{\vdots};
      \draw (8.7,0.1)node [color=black,font=\fontsize{10}{10}\selectfont]{\vdots};
      \draw (0.2,0)node [color=black,font=\fontsize{10}{10}\selectfont]{\vdots};
      \node[L2Node] (n5) at (7.7, 0){$F(z)$};
      \node[L2Node] (n6) at (9.7, 0){$G(z)$};
	   \draw[dashed](6.9,1)-- (10.5,1);
      \draw[dashed](6.9,1)-- (6.9,-1);
      \draw[dashed](6.9,-1)-- (10.5,-1);
      \draw[dashed](10.5,-1)-- (10.5,1);
      \draw[->](0,0.5)--(0.65,0.5);
      \draw(8.7,1)--(8.7,1.2);
      \draw(8.7,1.2)--(4.6,1.2);
      \draw[->](4.6,1.2)--(4.6,0.8);
      \draw[->](0,0.5)--(0.65,0.5);
      \draw[->](0,-0.5)--(0.65,-0.5);
      \draw[->](2.1,0.5)--(3.6,0.5);
      \draw[->](2.1,0)--(7,0);
      \draw[->](2.1,-0.5)--(7,-0.5);
      \draw[->](5.5,0.5)--(7,0.5);
      \draw[->](8.4,0.5)--(9,0.5);
      \draw[->](8.4,-0.5)--(9,-0.5);
      \draw[->](10.4,0.5)--(11.5,0.5);
      \draw[->](10.4,-0.5)--(11.5,-0.5);
      \draw (0,0.75)node [color=black,font=\fontsize{6}{6}\selectfont]{$r_{1}$};
      \draw (0,-0.75)node [color=black,font=\fontsize{6}{6}\selectfont]{$r_{m}$};
      \draw (3,0.75)node [color=black,font=\fontsize{6}{6}\selectfont]{$r_{1}',\ldots,r_{p}'$};
      \draw (3,0.18)node [color=black,font=\fontsize{6}{6}\selectfont]{$r_{p+1}'$};
      \draw (3,-0.7)node [color=black,font=\fontsize{6}{6}\selectfont]{$r_{m}'$};
        \draw (6.2,0.75)node [color=black,font=\fontsize{6}{6}\selectfont]{$v_{1},\ldots,v_{p}$};
      \draw (8.7,0.75)node [color=black,font=\fontsize{6}{6}\selectfont]{$u_{1}$};
      \draw (8.7,-0.75)node [color=black,font=\fontsize{6}{6}\selectfont]{$u_{m}$};
      \draw (11.2,0.75)node [color=black,font=\fontsize{6}{6}\selectfont]{$y_{1},\ldots,y_{p}$};
      \draw (11.3,-0.75)node [color=black,font=\fontsize{6}{6}\selectfont]{$\bar{y}_{p+1},\ldots,\bar{y}_m$};
      \draw (7,-1.4)node [color=black,font=\fontsize{10}{10}\selectfont]{\hspace{1cm} $W(z)$};
      \draw (5.5,1.35)node [color=black,font=\fontsize{6}{6}\selectfont]{$x_{1},\ldots,x_{p}$};
\end{tikzpicture}
    \caption{DRG-tf block diagram for non-sqaure systems with larger number of inputs. $\bar{y}_{p+1},\ldots,\bar{y}_m$ represent the outputs that are manually added to system $G(z)$ to transfer it into a square system.}
\label{fig:DRG-tf nonsquarelarger}
\end{figure}

Assume that system $G$ in Figure~\ref{fig:DRG-tf nonsquarelarger} has $m$ inputs and $p$ outputs, with $m > p$:
\begin{equation}\label{eq:G_nonsquare}
\renewcommand*{\arraystretch}{.5}
\begin{bmatrix}
Y_1(z)\\
\vdots\\
Y_p(z) \\
\end{bmatrix}
=\underbrace{\begin{bmatrix}
      G_{11}(z) & \ldots & G_{1m}(z)\\
      \vdots & \ddots & \vdots \\
      G_{p1}(z) & \ldots & G_{pm}(z)
     \end{bmatrix}}_{G} \begin{bmatrix}
U_1(z)\\
\vdots\\
U_p(z) \\
\vdots\\
U_m(z)
\end{bmatrix}
 \end{equation}
We transform $G(z)$ into a square system as follows. We manually introduce $m-p$ outputs, $\bar{Y}_{p+1},\ldots, \bar{Y}_{m}$, leading to the square system $\widetilde{G}$, described below:
\begin{equation}\label{eq:Gbar_nonsquare}
\renewcommand*{\arraystretch}{.5}
\begin{bmatrix}
\begin{array}{c}
Y_1(z)\\
\vdots\\
Y_p(z) \\
 \hdashline \\
\bar{Y}_{p+1}(z) \\
\vdots\\
\bar{Y}_m(z) \\
\end{array}
\end{bmatrix}
=\underbrace{\begin{bmatrix}
\begin{array}{l r}
\begin{matrix}
G_{c_{1,p}} & G_{c_{p+1,m}}
\end{matrix} \\
      \hdashline \\
      \begin{matrix}
      0_{m-p,p} & \bar{G}
      \end{matrix}
     \end{array}
     \end{bmatrix}}_{\widetilde{G}} \begin{bmatrix}
U_1(z)\\
\vdots\\
U_p(z) \\
\vdots\\
U_m(z)
\end{bmatrix}
 \end{equation}
where  $\bar{G}$ is an $(m-p) \times (m-p)$ transfer matrix representing the fictitious outputs, and $G_{c_{1,p}}$ and $G_{c_{p+1,m}}$ denote the first $p$ columns of $G$ and the last $(m-p)$ columns of $G$, respectively.
 
Note that the choice of the fake dynamics (i.e., $
[0_{m-p,p} \quad \bar{G}]
$) in \eqref{eq:Gbar_nonsquare} is not unique. The reason we use this structure of $[0_{m-p,p} \quad \bar{G}]$ is that  $\widetilde{G}^{-1}$ and $F$ can be easily obtained through block matrix inversion  (\cite{lu2002inverses}), and the structure of $F$ is easy to study, as will be explained below.  
For the diagonal method in DRG-tf, the decoupled system $W$ (see Figure~\ref{fig:DRG-tf nonsquarelarger}) is constructed as:
\begin{equation}\label{eq:Wz_nonsquare}
W= \begin{bmatrix}
\begin{array}{c:c}
\underbrace{\begin{matrix}
\renewcommand*{\arraystretch}{.5}
      G_{11}(z) & \ldots & 0\\
      \vdots & \ddots & \vdots \\
      0(z) & \ldots & G_{pp}(z)
      \end{matrix}}_{W_p} & 0_{p,(m-p)}\\
      \hdashline \\
      0_{(m-p),p} & \bar{G}_w
     \end{array}
     \end{bmatrix}
\end{equation}
where $\bar{G}_w$ is a $(m-p)\times(m-p)$ transfer function matrix that is chosen such that it has a stable inverse, so that $F^{-1}$ can be computed (see \eqref{eq:diagonal method}). 
Recall that the true outputs of the system are $Y_1,\ldots,Y_p$ and the constraints are on these outputs. So, as Figure~\ref{fig:DRG-tf nonsquarelarger} shows, only $p$ different SRGs are needed to ensure these outputs satisfy the constraints and there is no need to design SRGs for $\bar{G}_w$. Finally, $F^{-1}$ is introduced to ensure that $u$ is close to $r$, as before. By choosing $\widetilde{G}$ and $W$ as shown in \eqref{eq:Gbar_nonsquare} and \eqref{eq:Wz_nonsquare}, $F$ can be written as:
\begin{equation}\label{eq:Fz_nonsquare}
\renewcommand*{\arraystretch}{1.5}
F= \begin{bmatrix}
      G_{c_{1,p}}^{-1}W_p & -G_{c_{1,p}}^{-1}G_{c_{p+1,m}}\\
      0_{m-p,p} &  \bar{G}^{-1}\bar{G}_w 
     \end{bmatrix}
\end{equation}
Note that if we choose $\bar{G}$ to be equal to $\bar{G}_w$, then, $\bar{G}^{-1}\bar{G}_w$ in \eqref{eq:Fz_nonsquare} will become an identity matrix, which means that  $F$ is unrelated to the choice of $\bar{G}$. Of course, for this to hold, $\bar{G}$ needs to be invertible to ensure that \eqref{eq:Fz_nonsquare} exists.

\begin{remark} 
As can be seen from \eqref{eq:Fz_nonsquare}, if $\bar{G} \neq \bar{G}_w$, then $F$ is related to both $\bar{G}$ and $\bar{G}_w$. This implies that a proper set of $\bar{G}$ and $\bar{G}_w$ can be chosen such that the norm of $F$ is small, which as discussed in Section \ref{sec:analysisDRGtf}, will lead to a small distance between $u$ and $r$ (see Figure~\ref{fig:DRG-tf nonsquarelarger}) and, hence, good tracking performance.

\end{remark}
Since $G(z)$ has been transformed into a square system, the same analysis presented in Section \ref{sec:analysisDRGtf} can be applied to study the steady-state and transient performance of DRG-tf for non-square systems. Hence, we will not repeat this analysis. 

\subsection{Systems with larger number of outputs}

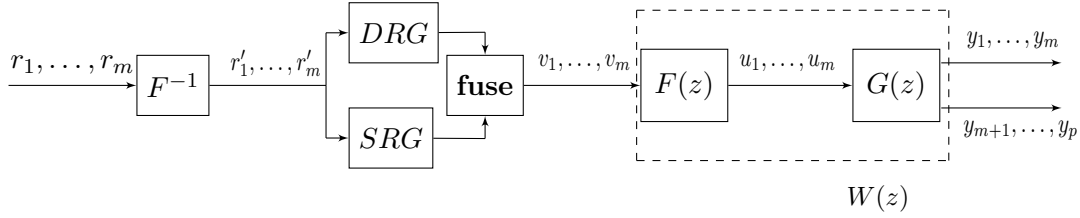
\begin{figure}
\centering
\begin{tikzpicture}
        [node distance=1.5cm,>=latex']
        \node [input, name=input] {};

        \node [block,minimum size=0.8cm, right=1.7cm of input] (Ginv) {$F^{-1}$};
        \node [output,minimum size=0.8cm, right=1.55cm of Ginv] (temp1){} ;
        \node [output,minimum size=0.8cm, above=0.7cm of temp1] (temp2){} ;
        \node [output,minimum size=0.8cm, below=0.7cm of temp1] (temp3){} ;
        \node [block,minimum size=0.8cm, right=0.3 of temp2] (RG) {$DRG$};
        \node [block,minimum size=0.8cm, right=0.3 of temp3] (SRG) {$SRG$};
        \node [block,minimum size=0.8cm, right =1.6 of temp1] (fuse) {$\textbf{fuse}$};
        \node [output,minimum size=0.8cm, above=0.3cm of fuse] (temp4){} ;
        \node [output,minimum size=0.8cm, below=0.3cm of fuse] (temp5){} ;
        \node [block, right=1.55 of fuse] (F) {$F(z)$};
        \node [block, right=1.65 of F] (G) {$G(z)$};
        \node [output, name=state1,below=0.3 of system] {};
        \node [output, name=state2,right=0.8 of state1] {};
        \node [output, below of=sum2] (loop) {};
        \node [output, name=state3,right=0.26 of feedback] {};
        \node [output, name=state4,above=0.6 of RG] {};
        \node [output, name=state5,below=0.31 of RG] {};
        \draw [draw,->] (input) -- node[above=0.2mm]{$r_1,\ldots,r_m$} (Ginv);
        \node [output, right=1 of G] (output) {};

        \draw [-] (Ginv) -- node [above,pos=0.79]{}(temp1);
        \draw [-] (temp1) -- node []{} (temp2);
        \draw [-] (temp1) -- node []{} (temp3);
        \draw [->] (temp2) -- node []{} (RG);
        \draw [->] (temp3) -- node []{} (SRG);
        \draw [-] (RG) -- node []{} (temp4);
        \draw [-] (SRG) -- node []{} (temp5);
        \draw [->] (fuse) -- node []{} (F);
        \draw [->] (temp4) -- node []{} (fuse);
        \draw [->] (temp5) -- node []{} (fuse);
        \draw [->] (F) -- node []{} (G);
 
      \draw (11,-1.5)node
      [color=black,font=\fontsize{10}{10}\selectfont]{\hspace{1cm} $W(z)$};

      \draw (3,0.3)node
      [color=black,font=\fontsize{10}{10}\selectfont]{\hspace{1cm} \scalebox{.8}[1.0]{$r_1',\ldots,r_m'$}};
    \draw (7.1,0.3)node
      [color=black,font=\fontsize{10}{10}\selectfont]{\hspace{1cm} \scalebox{.8}[1.0]{$v_1,\ldots,v_m$}};
      \draw (9.8,0.3)node
      [color=black,font=\fontsize{10}{10}\selectfont]{\hspace{1cm} \scalebox{.8}[1.0]{$u_1,\ldots,u_m$}};
    \draw (12.8,0.6)node
      [color=black,font=\fontsize{10}{10}\selectfont]{\hspace{1cm} \scalebox{.8}[1.0]{$y_1,\ldots,y_m$}};
    \draw (12.9,-0.6)node
      [color=black,font=\fontsize{10}{10}\selectfont]{\hspace{1cm} \scalebox{.8}[1.0]{$y_{m+1},\ldots,y_p$}};       

	   \draw[dashed](8.35,1)-- (12.5,1);
      \draw[dashed](8.35,1)-- (8.35,-1);
      \draw[dashed](8.35,-1)-- (12.5,-1);
      \draw[dashed](12.5,-1)-- (12.5,1);
      \draw[->](12.4,0.3)--(14,0.3);
    \draw[->](12.4,-0.3)--(14,-0.3);

\end{tikzpicture}
\caption{DRG-tf block diagram for non-sqaure systems with larger number of inputs.}
\label{fig: DRG-tf non-square large y}
\end{figure}

Assume system $G(z)$ in Figure~\ref{fig: DRG-tf non-square large y} has $m$ inputs and $p$ outputs, with $p > m$. Instead of decoupling the entire $G(z)$ as done in Section \ref{sec: drg-tf}, only a square subsystem of $G$ is decoupled. 
Without loss of generality, we assume that the square subsystem corresponds to the first $m$ outputs of $G$, but the method can be applied to other square subsystems as well. Let us denote the $m \times m$ square subsystem of $G$ as $G_m$. 
Same as DRG-tf for square systems (see Section \ref{sec:diag}),  $F$ is designed to decouple $G_m$, resulting in the diagonal subsystem, $W_m$,  shown below:
\begin{equation}\label{eq:Wm_nonsquare_large_y}
W_m= \begin{bmatrix}
G_{m_{11}}(z) & \ldots & 0\\
      \vdots & \ddots & \vdots \\
      0(z) & \ldots & G_{m_{mm}}(z)
      \end{bmatrix}
\end{equation}
Then, the whole system  $W$ (i.e., $GF$) can be described by:
\begin{equation}\label{eq:Wz_nonsquare_large_y}
W= \begin{bmatrix}
\begin{array}{c}
W_m\\
      \hdashline \\
     \underbrace{ FG_{m+1,p}}_{W_p} 
     \end{array}
     \end{bmatrix}
\end{equation}
where  $G_{m+1,p}$ represents the last $(p-m)$ rows of $G$.

As can be seen from Figure \ref{fig: DRG-tf non-square large y}, we  design one DRG (which contains $m$ decoupled SRGs) for $W_m$ to ensure that the outputs  $y_{1},\ldots,y_{m}$  satisfy the constraints. Then, we design a single SRG for $W_p$ to make sure that the outputs $y_{m+1},\ldots,y_{p}$ satisfy the constraints. 
The challenge is that two sets of $v$'s are computed: one by the DRG and one by the SRG (as shown in Figure \ref{fig: DRG-tf non-square large y}).  Thus, the question is, how can the two sets of $v$'s be ``fused" together while satisfying the constraints on all outputs. There are several ways to accomplish this task. The easiest solution is to select  the smallest $\kappa$ among the $m+1$  different $\kappa$'s ($\kappa$ is calculated based on \eqref{eq: LP for RG}), denoted as $\bar{\kappa}$, that is:
\begin{equation}\label{eq: minikappa}
\bar{\kappa}=\textbf{min}(\kappa_1,\ldots,\kappa_{m+1})
\end{equation}
and the update law for $v$ becomes: 
$$
v(t+1)=v(t)+\bar{\kappa}(r'(t)-v(t)).
$$
With the above  $\bar{\kappa}$, the convexity of the maximal admissible sets (MAS)  guarantees that the constraints for all outputs are satisfied and the solutions from the DRG and SRG are unified. However, the response of this approach may be conservative since the smallest $\kappa$ is chosen. An alternative way to fuse the $v$'s is as follows. 
First, denote the set of $v$'s given by the SRG (see Figure \ref{fig: DRG-tf non-square large y}) as $v_s$ and the set of $v$'s given by the DRG as $v_d$. We solve an RG-like LP (see \eqref{eq: LP for RG}) to find the point in $O_\infty^{W_{m}}$ that is  closest to $v_{s}$ (recall that $O_\infty^{W_{m}}$ refers to the MAS for $W_{m}$), denoted as $v_{t_1}$.  Similarly, we  solve another LP to find the closest point to $v_d$ in $O_\infty^{W_p}$, where $O_\infty^{W_p}$ represents the MAS for $W_p$, denoted as $v_{t_2}$.  
Note that  $v_{t_1}$ and $v_{t_2}$ are both constraint-admissible for all outputs since they are in $O_\infty^{W_p}$ and $O_\infty^{W_m}$ at the same time.

 Finally, we choose the actual set of $v$'s that is applied to $F(z)$ as:
$$
v=
    \begin{cases}
      v_{t_1} & \text{if $\|r'-v_{t_1}\| \leq \|r'-v_{t_2}\|$}\\
      v_{t_2} & \text{otherwise}\\
    \end{cases}  
$$
By choosing $v$ as above, it is guaranteed that the constraints for all outputs are satisfied. However, computational burden of this approach is higher than standard DRG since two more LPs are required. Finally, $F^{-1}$ is introduced to ensure that $u$ is close to $r$, as before.

\section{Conclusion}\label{Sec: Conclusions}

 In this work, a method for constraint management of coupled  MIMO systems was studied. The method is referred to as the Decoupled Reference Governor (DRG) and is based on decoupling the input-output dynamics, followed by application of scalar reference governors to each decoupled channel. We presented the DRG formulation with two different decoupling techniques based on transfer functions and state-space, and demonstrated the applicability of the method as a function of the singular values of the system and the decoupling matrix. Finally, we presented steady-state and transient analyses of the DRG and compared the computation time of DRG with VRG. It was shown that DRG can run faster than VRG by two orders of magnitude. Unknown disturbances and parametric uncertainties were also addressed.


Future work will explore modifications to DRG to ensure that the inputs to the closed-loop system (i.e., $u$ in Figure~\ref{fig:Decoupled with RG}) remain below the references (i.e., $r$). We will also explore DRG formulations that have the ability to recover from constraint violation, should unknown disturbances or observer errors push the system outside of the maximal admissible sets. 

\medskip

\printbibliography

@INPROCEEDINGS{Osorio_2018,
  author={J. {Osorio} and H. R. {Ossareh}},
  booktitle={2018 IEEE Conference on Control Technology and Applications (CCTA)}, 
  title={A Stochastic Approach to Maximal Output Admissible Sets and Reference Governors}, 
  year={2018},
  volume={},
  number={},
  pages={704-709},}

@inproceedings{Osorio_2019,
  author    = {Joycer Osorio and Mario Santillo and Julia Buckland and Mrdjan Jankovic, and
              Hamid R. Ossareh},
  title     = {A Reference Governor Approach towards Recovery from Constraint Violation},
  booktitle = {American Control Conference, {ACC} 2019,
              Philadelphia, USA, July 10-12, 2019},
  pages     = {},
  year      = {2019},
}

@article{kalabic_2015,
  title={Reference governors: Theoretical Extensions and Practical Applications.},
  author={Kalabic, Uros},
  year={2015}
}

@article{Kolmanovsky_1998,
  title={Theory and computation of disturbance invariant sets for discrete-time linear systems},
  author={Kolmanovsky, Ilya and Gilbert, Elmer G},
  journal={Mathematical problems in engineering},
  volume={4},
  number={4},
  pages={317--367},
  year={1998},
  publisher={Hindawi Publishing Corporation}
}

@article {Gilbert_1999,
author = {Gilbert, Elmer G. and Kolmanovsky, Ilya},
title = {Fast reference governors for systems with state and control constraints and disturbance inputs},
journal = {International Journal of Robust and Nonlinear Control},
volume = {9},
number = {15},
publisher = {John Wiley & Sons, Ltd.},
issn = {1099-1239},
pages = {1117--1141},
year = {1999},
}

@article{falb1967decoupling,
  journal={IEEE Transactions on Automatic Control}, 
  title={Decoupling in the design and synthesis of multivariable control systems},
  author={Falb, Peter L and Wolovich, William A},
  year={1967}
}

@article{Garelli_2006,
  title={Limiting interactions in decentralized control of MIMO systems},
  author={Garelli, F and Mantz, RJ and De Battista, H},
  journal={Journal of Process Control},
  volume={16},
  number={5},
  pages={473--483},
  year={2006},
  publisher={Elsevier}
}

@ARTICLE{Camponogara_2002, 
author={E. Camponogara and D. Jia and B. H. Krogh and S. Talukdar}, 
journal={IEEE Control Systems}, 
title={Distributed model predictive control}, 
year={2002}, 
volume={22}, 
number={1}, 
pages={44-52}, 
doi={10.1109/37.980246}, 
ISSN={1066-033X}, 
month={Feb},}

@INPROCEEDINGS{Wang_2003,
  author={ {Wenlin Wang} and D. E. {Rivera} and K. G. {Kempf}},
  booktitle={Proceedings of the 2003 American Control Conference, 2003.}, 
  title={Centralized model predictive control strategies for inventory management in semiconductor manufacturing supply chains}, 
  year={2003},
  volume={1},
  number={},
  pages={585-590 vol.1},}

@article{Elliott_2013,
  title={Decentralized model predictive control of a multi-evaporator air conditioning system},
  author={Elliott, Matthew S and Rasmussen, Bryan P},
  journal={Control Engineering Practice},
  volume={21},
  number={12},
  pages={1665--1677},
  year={2013},
  publisher={Elsevier}
}

@article{Bemporad_2002,
  title={Model predictive control based on linear programming\~{} the explicit solution},
  author={Bemporad, Alberto and Borrelli, Francesco and Morari, Manfred and others},
  journal={IEEE Transactions on Automatic Control},
  volume={47},
  number={12},
  pages={1974--1985},
  year={2002}
}

@book{Burl_1998,
  title={Linear optimal control: H (2) and H (Infinity) methods},
  author={Burl, Jeff B},
  year={1998},
  publisher={Addison-Wesley Longman Publishing Co., Inc.}
}

@ARTICLE{Ge_2014, 
  author={S. S. {Ge} and Z. {Li}},
  journal={IEEE Transactions on Automatic Control}, 
  title={Robust Adaptive Control for a Class of MIMO Nonlinear Systems by State and Output Feedback}, 
  year={2014},
  volume={59},
  number={6},
  pages={1624-1629},}

@inproceedings{Herceg_2013,
author = {Herceg, M and Kvasnica, M and Jones, C and Morari, M},
booktitle = {Proc. of the European Control Conference},
title = {{Multi-Parametric Toolbox 3.0}},
year = {2013}
}

@book{Zhou_1996,
  title={Robust and optimal control},
  author={Zhou, Kemin and Doyle, John Comstock and Glover, Keith and others},
  volume={40},
  year={1996},
  publisher={Prentice hall New Jersey}
}

@article{Scattolini_2009,
  title={Architectures for distributed and hierarchical model predictive control--a review},
  author={Scattolini, Riccardo},
  journal={Journal of process control},
  volume={19},
  number={5},
  pages={723--731},
  year={2009},
  publisher={Elsevier}
}

@book{Aastrom_1995,
  title={PID controllers: theory, design, and tuning},
  author={{\AA}str{\"o}m, Karl Johan and H{\"a}gglund, Tore},
  volume={2},
  year={1995},
  publisher={Instrument society of America Research Triangle Park, NC}
}

@incollection{Macfarlane_1983,
  title={A quasi-classical approach to multivariable feedback systems design},
  author={MacFarlane, AGJ and Hung, YS},
  booktitle={Computer Aided Design of Multivariable Technological Systems},
  pages={43--52},
  year={1983},
  publisher={Elsevier}
}

@ARTICLE{MacFarlane_1970,
  author={A. G. J. {MacFarlane}},
  journal={Electronics Letters}, 
  title={Commutative controller: a new technique for the design of multivariable control systems}, 
  year={1970},
  volume={6},
  number={5},
  pages={121-123},}

@book{Skogestad_2007,
%   title={Multivariable feedback control: analysis and design},
%   author={Skogestad, Sigurd and Postlethwaite, Ian},
%   volume={2},
%   year={2007},
%   publisher={Wiley New York}
% }

@INPROCEEDINGS{Ilya_1995,
  author={I. {Kolmanovsky} and E. G. {Gilbert}},
  booktitle={Proceedings of 1995 American Control Conference - ACC'95}, 
  title={Maximal output admissible sets for discrete-time systems with disturbance inputs}, 
  year={1995},
  volume={3},
  number={},
  pages={1995-1999 vol.3},}

@INPROCEEDINGS{Gilbert_1995,
  title={Discrete-time reference governors for systems with state and control constraints and disturbance inputs},
  author={Gilbert, Elmer G and Kolmanovsky, Ilya},
  booktitle={Decision and Control, 1995., Proceedings of the 34th IEEE Conference on},
  volume={2},
  pages={1189--1194},
  year={1995},
  organization={IEEE}
}

@inproceedings{Pluymers_2005,
  title={The efficient computation of polyhedral invariant sets for linear systems with polytopic uncertainty},
  author={Pluymers, B and Rossiter, JA and Suykens, JAK and De Moor, Bart},
  booktitle={Proceedings of the 2005, American Control Conference, 2005.},
  pages={804--809},
  year={2005},
  organization={IEEE}
}

@article{GARONE2017306,
  title={Reference and command governors for systems with constraints: A survey on theory and applications},
  author={Garone, Emanuele and Di Cairano, Stefano and Kolmanovsky, Ilya},
  journal={Automatica},
  volume={75},
  pages={306--328},
  year={2017},
  publisher={Elsevier}
}

@INPROCEEDINGS{Shah_2011,
  author={G. {Shah} and S. {Engell}},
  booktitle={Proceedings of the 2011 American Control Conference}, 
  title={Tuning MPC for desired closed-loop performance for MIMO systems}, 
  year={2011},
  volume={},
  number={},
  pages={4404-4409},}

@ARTICLE{Gilbert_1991,
  author={E. G. {Gilbert} and K. T. {Tan}},
  journal={IEEE Transactions on Automatic Control}, 
  title={Linear systems with state and control constraints: the theory and application of maximal output admissible sets}, 
  year={1991},
  volume={36},
  number={9},
  pages={1008-1020},}

@inproceedings{Kolmanovsky_2014,
  title={Reference and command governors: A tutorial on their theory and automotive applications},
  author={Kolmanovsky, Ilya and Garone, Emanuele and Di Cairano, Stefano},
  booktitle={American Control Conference (ACC), 2014},
  pages={226--241},
  year={2014},
  organization={IEEE}
}

@INPROCEEDINGS{DRG2018,
  author={Y. {Liu} and J. {Osorio} and H. {Ossareh}},
  booktitle={2018 IEEE Conference on Decision and Control (CDC)}, 
  title={Decoupled Reference Governors for Multi-Input Multi-Output Systems}, 
  year={2018},
  volume={},
  number={},
  pages={1839-1846},}

@phdthesis{kerrigan2001robust,
  title={Robust constraint satisfaction: Invariant sets and predictive control},
  author={Kerrigan, Eric Colin},
  year={2001},
  school={University of Cambridge}
}

@article{mcdonald1991ℓ1,
  title={L1-optimal control of multivariable systems with output norm constraints},
  author={McDonald, JS and Pearson, JB},
  journal={Automatica},
  volume={27},
  number={2},
  pages={317--329},
  year={1991},
  publisher={Elsevier}
}

@article{tee2009barrier,
  title={Barrier Lyapunov functions for the control of output-constrained nonlinear systems},
  author={Tee, Keng Peng and Ge, Shuzhi Sam and Tay, Eng Hock},
  journal={Automatica},
  volume={45},
  number={4},
  pages={918--927},
  year={2009},
  publisher={Elsevier}
}

@article{scokaert1998constrained,
  title={Constrained linear quadratic regulation},
  author={Scokaert, Pierre OM and Rawlings, James B},
  year={1998},
  journal={IEEE Transactions on Automatic Control},
  publisher={IEEE, Piscataway, NJ, USA}
}

@article{lloyd1970decoupling,
  title={Decoupling a multivariable discrete-time system},
  author={Lloyd, S},
  journal={Electronics Letters},
  volume={6},
  number={26},
  pages={831},
  year={1970},
  publisher={IET}
}

@article{silverman1970decoupling,
  title={Decoupling with state feedback and precompensation},
  author={Silverman, L},
  journal={IEEE Transactions on Automatic Control},
  volume={15},
  number={4},
  pages={487--489},
  year={1970},
  publisher={IEEE}
}

@article{nonlinearchen,
  title={Stability of nonlinear systems},
  author={Chen, Guanrong},
  journal={Encyclopedia of RF and Microwave Engineering},
  pages={4881--4896},
  year={2004},
  publisher={New York, USA: Wiley}
}

@book{harris1983stability,
  title={The stability of input-output dynamical systems},
  author={Harris, Christopher John and Valenca, JME},
  volume={168}
}

@article{lu2002inverses,
  title={Inverses of 2x2 block matrices},
  author={Lu, Tzon-Tzer and Shiou, Sheng-Hua},
  journal={Computers and Mathematics with Applications},
  volume={43},
  number={1-2},
  pages={119--129},
  year={2002},
  publisher={Elsevier}
}

@article{tondel2003algorithm,
  title={An algorithm for multi-parametric quadratic programming and explicit MPC solutions},
  author={T{\o}Ndel, Petter and Johansen, Tor Arne and Bemporad, Alberto},
  journal={Automatica},
  volume={39},
  number={3},
  pages={489--497},
  year={2003},
  publisher={Elsevier}
}
\end{document}